\documentclass[submission,copyright]{eptcs}

\providecommand{\event}{GandALF 2024}
\title{Reachability and Safety Games under TSO Semantics \\ \normalsize{(Extended Version)}}
\author{
    Stephan Spengler
    \institute{Uppsala University \\ Uppsala, Sweden}
    \email{stephan.spengler@it.uu.se}
}
\def\titlerunning{Reachability and Safety Games under TSO Semantics}
\def\authorrunning{S. Spengler}

\usepackage{aliascnt}
\usepackage{amsmath}
\usepackage{amssymb}
\usepackage{amsthm}
\usepackage{breakurl}
\usepackage[capitalise, noabbrev]{cleveref}
\usepackage{booktabs}
\usepackage{enumerate}
\usepackage{float}
\usepackage[T1]{fontenc}
\usepackage{hyperref}
\usepackage{multirow}
\usepackage{subcaption}
\usepackage{tabularx}
\usepackage{textcomp}
\usepackage{tikz}
\usepackage{underscore}
\usepackage{xcolor}

\DeclareMathAlphabet{\mathcal}{OMS}{cmsy}{m}{n}
\newcolumntype{C}[1]{>{\centering\arraybackslash}p{#1}}
\usetikzlibrary{decorations.pathreplacing,shapes}
\newcommand{\multicaption}[2]{\caption{\tabular[t]{@{}l@{}}#1\\#2\endtabular}}

\newtheorem{theorem}{Theorem}

\newaliascnt{lemma}{theorem}
\newtheorem{lemma}[lemma]{Lemma}
\aliascntresetthe{lemma}

\newaliascnt{corollary}{theorem}

\aliascntresetthe{corollary}

\newaliascnt{claim}{theorem}
\newtheorem{claim}[claim]{Claim}
\aliascntresetthe{claim}

\theoremstyle{definition} 
\newaliascnt{remark}{theorem}

\aliascntresetthe{remark}

\newtheorem*{lemma*}{Lemma}
\newtheorem*{claim*}{Claim}

\newcommand{\inference}[3]{\textbf{#1} & \frac{#2}{#3}}

\newcommand{\Nat}{\mathbb{N}}

\newcommand{\newsemantics}[2]{\newcommand{#1}{\mathsf{#2}}}
\newcommand{\newfunction}[2]{\newcommand{#1}{\mathop\mathrm{#2}}} 
\newcommand{\newcomponent}[2]{\newcommand{#1}{\mathsf{#2}}}
\newcommand{\newinstruction}[2]{\newcommand{#1}{\mathtt{#2}}}
\newcommand{\renewinstruction}[2]{\renewcommand{#1}{\mathtt{#2}}}

\newcommand{\sizeof}[1]{|#1|}

\newcommand{\of}[1]{(#1)}
\newcommand{\tuple}[1]{\langle#1\rangle}
\newcommand{\set}[1]{\{#1\}}

\newcommand{\cof}[1][]{\ifthenelse{\isempty{#1}}{}{\of{#1}}}
\renewcommand{\to}[1][]{\mathop{\xrightarrow{~#1~}}}
\renewcommand{\part}{\rightharpoonup}
\newcomponent{\word}{w}

\newcommand{\newclass}[2]{\newcommand{#1}{\textsc{#2}}}
\newclass{\exptime}{ExpTime}
\newclass{\etime}{ETime}
\newclass{\nexptime}{NexpTime}
\newclass{\class}{Class}
\newclass{\pspace}{PSpace}
\newclass{\npspace}{NPSpace}
\newclass{\expspace}{ExpSpace}
\newclass{\dtime}{DTime}
\newclass{\dspace}{DSpace}

\newcommand{\TS}{\mathcal{T}}
\newcommand{\kstar}{^{*}}

\newcomponent{\conf}{c}
\newcomponent{\confset}{C}
\newcomponent{\lbl}{label}
\newcomponent{\lblset}{L}
\newcomponent{\state}{q}
\newcomponent{\stateset}{Q}

\newfunction{\post}{Post}
\newfunction{\pre}{Pre}

\newsemantics{\final}{final}
\newsemantics{\target}{target}
\newsemantics{\all}{all}
\newsemantics{\true}{true}
\newsemantics{\false}{false}

\newcommand{\program}{\mathcal{P}}
\newcomponent{\process}{Proc}
\newcomponent{\tso}{TSO}
\newcomponent{\dtso}{DTSO}
\newcomponent{\transition}{\delta}

\newcomponent{\xvar}{x}
\newcomponent{\varset}{Vars}
\newcomponent{\dval}{d}
\newcomponent{\valset}{Dom}
\newcomponent{\msg}{m}
\newcomponent{\instr}{instr}
\newcomponent{\instrs}{Instrs}
\newcommand{\xd}{\xvar, \dval}
\newcomponent{\self}{self}
\newcomponent{\other}{other}
\newcomponent{\own}{own}

\newinstruction{\rd}{rd}
\renewinstruction{\wr}{wr}
\newinstruction{\nop}{skip}
\newinstruction{\mf}{mf}
\newinstruction{\up}{up} 
\newinstruction{\prop}{prop} 
\newinstruction{\del}{del} 

\newcommand{\indexset}{\mathcal{I}}
\newcommand{\statemap}{\mathcal{S}}
\newcommand{\buffermap}{\mathcal{B}}
\newcommand{\memorymap}{\mathcal{M}}
\newcommand{\pid}{\iota}

\newcomponent{\view}{v}
\newcomponent{\viewset}{V}
\newcommand{\valuemap}{\mathcal{V}}
\newcommand{\fencemap}{\mathcal{F}}

\newcommand{\game}{\mathcal{G}}
\newcomponent{\play}{P}
\newcomponent{\wincon}{W}
\newcomponent{\bisim}{Z}
\newcommand{\colset}{\mathcal{C}}
\newcommand{\extend}{\uparrow_\program}
\newcommand{\restrict}{\downarrow^\pid}
\newcomponent{\subconfset}{D}
\newcomponent{\subconf}{d}
\newcomponent{\altconf}{d}

\newcommand{\channelsystem}{\mathcal{L}}
\newcomponent{\channelstate}{s}
\newcomponent{\channelstateset}{S}
\newcomponent{\channelset}{L}
\newcomponent{\channelmessage}{m}
\newcomponent{\channelmessageset}{M}
\newcomponent{\channeloperation}{op}
\newcomponent{\channeloperationset}{Op}

\newcomponent{\letter}{\sigma}
\newcomponent{\alphabet}{\Sigma}
\newcomponent{\tmL}{L}
\newcomponent{\tmR}{R}
\newcomponent{\tmD}{D}
\newcomponent{\pos}{i}
\newcomponent{\posj}{j}

\newcommand{\xrd}{\xvar_\rd}
\newcommand{\xwr}{\xvar_\wr}
\newcomponent{\yvar}{y}
\newcomponent{\zvar}{z}

\newcomponent{\hstate}{h}
\newcomponent{\rstate}{r}

\newcomponent{\channeledge}{e}

\begin{document}

\maketitle

\begin{abstract}
We consider games played on the transition graph of concurrent programs running under the Total Store Order (TSO) weak memory model.
Games are frequently used to model the interaction between a system and its environment, in this case between the concurrent processes and the nondeterministic TSO buffer updates.
In our formulation, the game is played by two players, who alternatingly make a move:
The \emph{process player} can execute any enabled instruction of the processes, while the \emph{update player} takes care of updating the messages in the buffers that are between each process and the shared memory.
We show that the reachability and safety problem of this game reduce to the analysis of single-process (non-concurrent) programs.
In particular, they exhibit only finite-state behaviour.
Because of this, we introduce different notions of \emph{fairness}, which force the two players to behave in a more realistic way.
Both the reachability and safety problem then become undecidable.
\end{abstract}

\section{Introduction}

In concurrent programs, different processes interact with each other through the use of shared memory. Programmers usually unconsciously assume that the semantics adhere to the Sequential Consistency (SC) memory model \cite{DBLP:journals/tc/Lamport79}. In SC, the execution of processes can be interleaved, but write instructions are visible in the memory in the exact order in which they were issued. However, most modern architectures, such as Intel x86 \cite{x86-swdmanual-1-3}, SPARC \cite{sparc9}, IBM's POWER \cite{power-isa-v31b}, and ARM \cite{arm-v7ar-refman}, implement several relaxations and optimisations that improve memory access latency but break SC assumptions. A standard model that is weaker than SC allows the reordering of reads and writes of the same process, as long as it maintains the appearance of SC from the perspective of each individual process. The implementation of this optimisation adds an unbounded first-in-first-out write buffer between each process and the shared memory. The buffer is used to delay write operations. This model is called Total Store Ordering (TSO) and is a faithful formalisation of SPARC and Intel x86 ~\cite{DBLP:conf/tphol/OwensSS09,DBLP:journals/cacm/SewellSONM10}.

Verification under TSO semantics is difficult due to the unboundedness of the buffers. Even if each process can be modelled as a finite-state system, the program itself has a state space of infinite size. The reachability problem for programs running under TSO semantics is to decide whether a target program state is reachable from a given initial state during program execution. If the target state is considered to be a bad state, it is also called the safety problem. Solving reachability and safety helps in deciding if a program is correct, i.e. if it adheres to a specification or if it can avoid states of undefined behaviour. Using alternative but equivalent semantics, it has been shown that the reachability problem is decidable \cite{DBLP:conf/popl/AtigBBM10,DBLP:conf/tacas/AbdullaACLR12,DBLP:journals/lmcs/AbdullaABN18}. Furthermore, lossy channel system \cite{wsts2,wsts1,DBLP:conf/icalp/AbdullaJ94,DBLP:journals/ipl/Schnoebelen02} can be simulated by programs running under TSO semantics \cite{DBLP:conf/popl/AtigBBM10}. This implies that the reachability problem is non-primitive recursive \cite{DBLP:journals/ipl/Schnoebelen02} and the repeated reachability problem is undecidable \cite{DBLP:conf/icalp/AbdullaJ94}. Additionally, the termination problem has been shown to be decidable \cite{DBLP:journals/siglog/Atig20} using the framework of well-structured transition systems \cite{wsts1,wsts2}.

In this paper, we consider games played on the transition graph of concurrent programs running under TSO semantics. Formal games provide a framework to reason about the behaviour of a system and the interaction between the system and its environment. In particular, they have been extensively used in controller synthesis problems \cite{DBLP:journals/tcs/ArnoldVW03,DBLP:conf/ecbs/BackS04,DBLP:conf/concur/BassetKW14,DBLP:conf/adhs/BalkanVT15,DBLP:journals/corr/abs-2209-10319}. A previous paper introduces safety games in which two players alternatingly execute instructions of a concurrent program \cite{DBLP:journals/corr/abs-2310-00990}. Motivated by this work, we propose a game setting that more closely models the interplay between a system and the environment: The first player controls the execution of the program instructions, while the second player handles the nondeterministic updates of the store buffers to the shared memory. This model sees the process and the update mechanism as antagonistic, and allows us to reason about the correctness of the program regardless of the update behaviour.

We consider two types of game objectives: In a reachability game, the process player tries to reach a given set of target states, while the update player tries to avoid this; In a safety game, these two roles are reversed. We show that in both cases finding the winner of the game reduces to the analysis of games being played on a program with just one process. Furthermore, we show that these games are bisimilar to finite-state games and thus decidable. In particular, the reachability and safety problem are \pspace-complete. The reason that the concurrent programs exhibit a finite-state character lies in the optimal behaviour of the two players. If the player that controls the processes has a winning strategy, then she can win by playing in only one process, ignoring all the other processes of the program. On the other hand, if the player controlling the buffer is able to win, she can do so by never letting any write operation reach the memory. In both cases, there is no concurrency in the sense that the processes do not interact or communicate with each other. This is not realistic, since we should be able to assume that if the program runs a sufficiently long duration (1) every process will be executed and (2) every write stored in the buffer will be updated to the memory.

We rectify this issue by introducing two fairness conditions. First, in an infinite run the process player must execute each enabled process infinitely many times. Second, the update player must make sure that each write operation reaches the memory after finitely many steps. We show that both the reachability and safety problem become undecidable with these restrictions. To do so, we use a reduction from perfect channel systems adapted from \cite{DBLP:journals/corr/abs-2310-00990}.

Finally, we investigate an alternative TSO semantics in our game setting. The authors of \cite{DBLP:journals/lmcs/AbdullaABN18} propose a load-buffer semantics for TSO which reverts the direction of the information flow between the processes and the shared memory. In their model, the buffer is filled with values from the memory which can later be read by the process. Using well-structured transition systems, they showed that it is equivalent to the classical TSO semantics with respect to state reachability. We explore whether the equivalence also holds in the two-player game, but come to the conclusion that this is not the case. In particular, we construct a concurrent program that is won by the update player under store-buffer semantics but by the process player under load-buffer semantics.

\bigskip

This document is an extended version of the conference paper with the same title. It includes the full proofs of \autoref{lem:extend} and \autoref{claim:single-process} and gives a detailed explanation of the complexity of single-process games in \autoref{sec:complexity}.

\section{Preliminaries}

\paragraph{Transition Systems}
A \emph{(labelled) transition system} is a triple $\TS = \tuple{ \confset, \lblset, \to }$, where $\confset$ is a set of \emph{configurations}, $\lblset$ is a set of \emph{labels}, and $\to \subseteq \confset \times \lblset \times \confset$ is a \emph{transition relation}.
We usually write $\conf_1 \to[\lbl] \conf_2$ if $\tuple{ \conf_1, \lbl, \conf_2} \in \to$.
Furthermore, we write $\conf_1 \to \conf_2$ if there exists some $\lbl$ such that $\conf_1 \to[\lbl] \conf_2$.
A \emph{run} $\pi$ of $\TS$ is a sequence of transitions $\conf_0 \to[\lbl_1] \conf_1 \to[\lbl_2] \conf_2 \dots \to[\lbl_n] \conf_n$.
It is also written as $\conf_0 \to[\pi] \conf_n$.
A configuration $\conf'$ is \emph{reachable} from a configuration $\conf$, if there exists a run from $\conf$ to $\conf'$.

For a configuration $\conf$, we define $\pre\of\conf = \set{ \conf' \mid \conf' \to \conf }$ and $\post\of\conf = \set{ \conf' \mid \conf \to \conf' }$.
We extend these notions to sets of configurations $\confset'$ with $\pre(\confset') = \bigcup_{\conf \in \confset'} \pre\of\conf$ and $\post(\confset') = \bigcup_{\conf \in \confset'} \post\of\conf$.

An \emph{unlabelled transition system} is a transition system without labels.
Formally, it is defined as a transition system with a singleton label set.
In this case, we omit the labels.

\paragraph{Perfect Channel Systems}
Given a set of messages $\channelmessageset$, define the set of channel operations $\channeloperationset = \set{ !\channelmessage, ?\channelmessage \mid \channelmessage \in \channelmessageset} \cup \set\nop$.
A \emph{perfect channel system} (PCS) is a triple $\channelsystem = \tuple{ \channelstateset, \channelmessageset, \transition }$, where $\channelstateset$ is a set of states, $\channelmessageset$ is a set of messages, and $\transition \subseteq \channelstateset \times \channeloperationset \times \channelstateset$ is a transition relation.
We write $\channelstate_1 \to[\channeloperation] \channelstate_2$ if $\tuple{ \channelstate_1, \channeloperation, \channelstate_2 } \in \transition$.

Intuitively, a PCS models a finite state automaton that is augmented by a \emph{perfect} (i.e. non-lossy) FIFO buffer, called \emph{channel}.
During a \emph{send operation} $!\channelmessage$, the channel system appends $\channelmessage$ to the tail of the channel.
A transition $?\channelmessage$ is called \emph{receive operation}.
It is only enabled if the channel is not empty and $\channelmessage$ is its oldest message.
When the channel system performs this operation, it removes $\channelmessage$ from the head of the channel.
Lastly, a $\nop$ operation just changes the state, but does not modify the buffer.

The formal semantics of $\channelsystem$ are defined by a transition system $\TS_\channelsystem = \tuple{ \confset_\channelsystem, \lblset_\channelsystem, \to_\channelsystem }$, where $\confset_\channelsystem = \channelstateset \times \channelmessageset\kstar$, $\lblset_\channelsystem = \channeloperationset$ and the transition relation $\to_\channelsystem$ is the smallest relation given by:
\begin{itemize}
	\item If $\channelstate_1 \to[!\channelmessage] \channelstate_2$ and $\word \in \channelmessageset\kstar$, then $\tuple{ \channelstate_1, \word } \to[!\channelmessage]_\channelsystem \tuple{ \channelstate_2, \channelmessage \cdot \word }$.
	\item If $\channelstate_1 \to[?\channelmessage] \channelstate_2$ and $\word \in \channelmessageset\kstar$, then $\tuple{ \channelstate_1, \word \cdot \channelmessage } \to[?\channelmessage]_\channelsystem \tuple{ \channelstate_2, \word }$.
	\item If $\channelstate_1 \to[\nop] \channelstate_2$ and $\word \in \channelmessageset\kstar$, then $\tuple{ \channelstate_1, \word } \to[\nop]_\channelsystem \tuple{ \channelstate_2, \word }$.
\end{itemize}
A state $\channelstate_F \in \channelstateset$ is \emph{reachable} from a configuration $\conf_0 \in \confset_\channelsystem$, if there exists a configuration $\conf_F = \tuple{ \channelstate_F, \word_F }$ such that $\conf_F$ is reachable from $\conf_0$ in $\TS_\channelsystem$.
The \textbf{state reachability problem} of PCS is, given a perfect channel system $\channelsystem$, an initial configuration $\conf_0 \in \confset_\channelsystem$ and a final state $\channelstate_F \in \channelstateset$, to decide whether $\channelstate_F$ is reachable from $\conf_0$ in $\TS_\channelsystem$.
It is undecidable \cite{DBLP:journals/jacm/BrandZ83}.

\section{Concurrent Programs}

\paragraph{Syntax}
Let $\valset$ be a finite data domain and $\varset$ be a finite set of shared variables over $\valset$.
We define the \emph{instruction set} $\instrs = \set{ \rd\of\xd, \wr\of\xd \mid \xvar \in \varset, \dval \in \valset } \cup \set{ \nop, \mf }$,
which are called \emph{read}, \emph{write}, \emph{skip} and \emph{memory fence}, respectively.
A process is represented by a finite state labelled transition system.
It is given as the triple $\process = \tuple{ \stateset, \instrs, \transition }$, where $\stateset$ is a finite set of \emph{local states} and $\transition \subseteq \stateset \times \instrs \times \stateset$ is the transition relation.
As with transition systems, we write $\state_1 \to[\instr] \state_2$ if $\tuple{ \state_1, \instr, \state_2} \in \transition$ and $\state_1 \to \state_2$ if there exists some $\instr$ such that $\state_1 \to[\instr] \state_2$.

A \emph{concurrent program} is a tuple of processes $\program = \tuple{ \process^\pid }_{\pid \in \indexset}$, where $\indexset$ is a finite set of process identifiers.
For each $\pid \in \indexset$ we have $\process^\pid = \tuple{ \stateset^\pid, \instrs, \transition^\pid }$.
A \emph{global} state of $\program$ is a function $\statemap: \indexset \to \bigcup_{\pid \in \indexset} \stateset^\pid$ that maps each process to its local state, i.e $\statemap(\pid) \in \stateset^\pid$.

\paragraph{TSO Semantics}
Under TSO semantics, the processes of a concurrent program do not interact with the shared memory directly, but indirectly through a FIFO \emph{store buffer} instead.
When performing a \emph{write} instruction $\wr\of\xd$, the process adds a new message $\tuple\xd$ to the tail of its store buffer.
A \emph{read} instruction $\rd\of\xd$ works differently depending on the current buffer content of the process.
If the buffer contains a write message on variable $\xvar$, the value $\dval$ must correspond to the value of the most recent such message.
Otherwise, the value is read directly from memory.
A \emph{skip} instruction only changes the local state of the process.
The \emph{memory fence} instruction is disabled, i.e. it cannot be executed, unless the buffer of the process is empty.
Additionally, at any point during the execution, the process can \emph{update} the write message at the head of its buffer to the memory.
For example, if the oldest message in the buffer is $\tuple\xd$, it will be removed from the buffer and the memory value of variable $\xvar$ will be updated to contain the value $\dval$.
This happens in a nondeterministic manner.

Formally, we introduce a TSO \emph{configuration} as a tuple $\conf = \tuple{ \statemap, \buffermap, \memorymap }$, where:
\begin{itemize}
	\item $\statemap: \indexset \to \bigcup_{\pid \in \indexset} \stateset^\pid$ is a global state of $\program$.
	\item $\buffermap: \indexset \to (\varset \times \valset)\kstar$ represents the buffer state of each process.
	\item $\memorymap: \varset \to \valset$ represents the memory state of each shared variable.
\end{itemize}
Given a configuration $\conf$, we write $\statemap\of\conf$, $\buffermap\of\conf$ and $\memorymap\of\conf$ for the global program state, buffer state and memory state of $\conf$.
The semantics of a concurrent program running under TSO is defined by a transition system $\TS_\program = \tuple{ \confset_\program, \lblset_\program, \to_\program }$,
where $\confset_\program$ is the set of all possible TSO configurations
and $\lblset_\program = \set{ \instr_\pid \mid \instr \in \instrs, \pid \in \indexset } \cup \set{ \up_\iota \mid \pid \in \indexset }$ is the set of labels.
The transition relation $\to_\program$ is given by the rules in \autoref{fig:tso-semantics}, where we use $\buffermap\of\pid|_{\set\xvar \times \valset}$ to denote the restriction of $\buffermap\of\pid$ to write messages on the variable $\xvar$.
Furthermore, we define $\up\kstar$ to be the transitive closure of $\set{ \up_\iota \mid \pid \in \indexset }$, i.e. $\conf_1 \to[\up\kstar]_\program \conf_2$ if and only if $\conf_2$ can be obtained from $\conf_1$ by some amount of buffer updates.

\begin{figure}
\centering
\begin{equation*}
\begin{array}{lc}

\inference{read-own-write}
	{\state \to[\rd\of\xd] \state' \qquad \statemap\of\pid = \state \qquad \buffermap\of\pid|_{\set\xvar \times \valset} = \tuple\xd \cdot \word}
	{\tuple{ \statemap, \buffermap, \memorymap} \to[\rd\of\xd_\pid]_\program \tuple{ \statemap[\pid \leftarrow \state'], \buffermap, \memorymap}}
\bigskip\\
\inference{read-from-memory}
	{\state \to[\rd\of\xd] \state' \qquad \statemap\of\pid = \state \qquad \buffermap\of\pid|_{\set\xvar \times \valset} = \varepsilon \qquad \memorymap\of\xvar = \dval}
	{\tuple{ \statemap, \buffermap, \memorymap} \to[\rd\of\xd_\pid]_\program \tuple{ \statemap[\pid \leftarrow \state'], \buffermap, \memorymap}}
\bigskip\\
\inference{write}
	{\state \to[\wr\of\xd] \state' \qquad \statemap\of\pid = \state}
	{\tuple{ \statemap, \buffermap, \memorymap} \to[\wr\of\xd_\pid]_\program \tuple{ \statemap[\pid \leftarrow \state'], \buffermap[\pid \leftarrow \tuple\xd \cdot \buffermap\of\pid], \memorymap}}
\bigskip\\
\inference{skip}
	{\state \to[\nop] \state' \qquad \statemap\of\pid = \state}
	{\tuple{ \statemap, \buffermap, \memorymap} \to[\nop_\pid]_\program \tuple{ \statemap[\pid \leftarrow \state'], \buffermap, \memorymap}}
\bigskip\\
\inference{memory-fence}
	{\state \to[\mf] \state' \qquad \statemap\of\pid = \state \qquad \buffermap\of\pid = \varepsilon}
	{\tuple{ \statemap, \buffermap, \memorymap} \to[\mf_\pid]_\program \tuple{ \statemap[\pid \leftarrow \state'], \buffermap, \memorymap}}
\bigskip\\
\inference{update}
	{\buffermap\of\pid = \word \cdot \tuple\xd}
	{\tuple{ \statemap, \buffermap, \memorymap} \to[\up_\pid]_\program \tuple{ \statemap, \buffermap[\pid \leftarrow \word], \memorymap[\xvar \leftarrow \dval]}}

\end{array}
\end{equation*}
\caption{TSO semantics}
\label{fig:tso-semantics}
\end{figure}

A global state $\statemap_F$ is \emph{reachable} from an initial configuration $\conf_0$, if there is a configuration $\conf_F$ with $\statemap(\conf_F) = \statemap_F$ such that $\conf_F$ is reachable from $\conf_0$ in $\TS_\program$.
The \textbf{state reachability problem} of TSO is, given a program $\program$, an initial configuration $\conf_0$ and a final global state $\statemap_F$, to decide whether $\statemap_F$ is reachable from $\conf_0$ in $\TS_\program$.

\section{Games}


A \emph{game} is an unlabelled transition system, in which two players A and B take turns making a \emph{move} in the transition system, i.e. changing the state of the game from one configuration to an adjacent one.
In a \emph{reachability game}, the goal of player A is to reach a given set of target states, while player B tries to avoid this.
In a \emph{safety game}, the roles are swapped.

Formally, a game is defined as a tuple $\game = \tuple{ \confset, \confset_A, \confset_B, \to}$, where $\confset$ is the set of configurations, $\confset_A$ and $\confset_B$ form a partition of $\confset$, and $\to$ is a transition relation on $\confset$.
For the games considered in this paper, the relation will always be restricted to $\to \subseteq (\confset_A \times \confset_B) \cup (\confset_B \times \confset_A)$, which means that the two players take turns alternatingly.

\paragraph{Plays and Winning Conditions}
An \emph{infinite play} $\play$ of $\game$ is an infinite sequence $\conf_0, \conf_1, \dots$ such that $\conf_i \to \conf_{i+1}$ for all $i \in \Nat$.
Similarly, a \emph{finite play} is a finite sequence $\conf_0, \conf_1, \dots, \conf_n$ such that $\conf_i \to \conf_{i+1}$ for all $i \in [0, \dots, n-1]$ and $\post(\conf_n) = \emptyset$, i.e. the play ends in a deadlock.
A \emph{winning condition} $\wincon$ is a subset of all infinite plays.
We say that player A is the winner of a play, if either the play is infinite and an element of $\wincon$, or if it is finite and ends in a deadlock for player B, i.e. $\conf_n \in \confset_B$.
Otherwise, player B wins the play.

In this work, we will consider two types of winning conditions.
A \emph{reachability condition} is given by a set $\confset_R \subseteq \confset$ which induces the winning condition $\wincon_R = \set{ \play = \conf_0, \conf_1, \dots \mid \exists i \in \Nat: \conf_i \in \confset_R}$, i.e. the set of all plays that visit a configuration in $\confset_R$.
Accordingly, a \emph{safety condition} is given by a set $\confset_S \subseteq \confset$ which induces the winning condition $\wincon_S = \set{ \play = \conf_0, \conf_1, \dots \mid \forall i \in \Nat: \conf_i \not\in \confset_S}$, i.e. the set of all plays that never visit a configuration in $\confset_S$.
Reachability games and safety games are dual to each other in the sense that a reachability game with winning condition $\confset_R$ can be seen as a safety game with winning condition $\confset_S = \confset \setminus \confset_R$, where the roles of players A and B are swapped.

\paragraph{Strategies}
A \emph{strategy} of player A is a partial function $\sigma_A: \confset\kstar \part \confset_B$, such that $\sigma_A(\conf_0, \dots, \conf_n)$ is defined if and only if $\conf_0, \dots, \conf_n$ is a prefix of a play, $\conf_n \in \confset_A$ and $\sigma_A(\conf_0, \dots, \conf_n) \in \post(\conf_n)$.
A strategy $\sigma_A$ is called \emph{positional}, if it only depends on $\conf_n$, i.e. if $\sigma_A(\conf_0, \dots, \conf_n) = \sigma_A(\conf_n)$ for all $(\conf_0, \dots, \conf_n)$ on which $\sigma_A$ is defined.
Thus, a positional strategy is usually given as a total function $\sigma_A: \confset_A \to \confset_B$.
For player B, strategies are defined accordingly.

Two strategies $\sigma_A$ and $\sigma_B$ together with an initial configuration $\conf_0$ induce a finite or infinite play $\play(\conf_0, \sigma_A, \sigma_B) = \conf_0, \conf_1, \dots$ such that $\conf_{i+1} = \sigma_A(\conf_0, \dots, \conf_i)$ for all $\conf_i \in \confset_A$ and $\conf_{i+1} = \sigma_B(\conf_0, \dots, \conf_i)$ for all $\conf_i \in \confset_B$.
Given a winning condition $\wincon$, a strategy $\sigma_A$ is \emph{winning} from a configuration $\conf_0$, if for \emph{all} strategies $\sigma_B$ it holds that player A wins the play $\play(\conf_0, \sigma_A, \sigma_B)$.
That is, for each $\sigma_B$, either $\play(\conf_0, \sigma_A, \sigma_B) \in \wincon$ or the play is finite and ends in a deadlock of player B.
A configuration $\conf_0$ is \emph{winning} for player A if she has a strategy that is winning from $\conf_0$.
Equivalent notions exist for player B.
Given a reachability condition $\wincon_R$ / a safety condition $\wincon_S$, the \textbf{reachability problem} / \textbf{safety problem} for a game $\game$ and a configuration $\conf_0$ is to decide whether $\conf_0$ is winning for player A.
\begin{lemma}[Proposition 2.21 in \cite{DBLP:conf/dagstuhl/Mazala01}]
\label{lem:positional}
    In reachability and safety games, every configuration is winning for exactly one player.
    A player with a winning strategy also has a positional winning strategy.
\end{lemma}
Since we only consider reachability and safety games in this paper, all strategies will be positional.

\section{Reachability and Safety Games under TSO Semantics}
\label{sec:tso-games}

We model the execution of a TSO program as a game between two players:
The \emph{process player A} takes the role of the program and decides at each execution step which instruction to execute.
The \emph{update player B} is in charge of the buffer message updates in between.

Formally, a TSO program $\program = \tuple{\process^\pid}_{\pid \in \indexset}$ induces a game $\game(\program) = \tuple{ \confset, \confset_A, \confset_B, \to}$ as follows.
The sets $\confset_A$ and $\confset_B$ are copies of the set $\confset^\program$ of TSO configurations, annotated by $A$ and $B$, respectively: $\confset_A := \set{ \conf_A \mid \conf \in \confset^\program}$ and $\confset_B := \set{ \conf_B \mid \conf \in \confset^\program}$.
The transition relation $\to$ is defined by the following rules:
\begin{itemize}
    \item
    \textbf{Program}\quad
    For each transition $\conf \to[\instr_\pid]_\program \conf'$ where $\conf, \conf' \in \confset^\program$, $\pid \in \indexset$ and $\instr \in \instrs$, it holds that $\conf_A \to[\instr_\pid] \conf'_B$.
    This means that the process player can execute any program instruction.
    \item
    \textbf{Update}\quad
    For each transition $\conf_B \in \confset_B$, it holds that $\conf_B \to[\up\kstar] \conf'_A$ for all $\conf'$ with $\conf \to[\up\kstar]_\program \conf'$.
    This means that the update player can update any amount of buffer messages (including zero) between each of the turns of the process player.
\end{itemize}

In the remainder of this work, we will consider reachability or safety winning conditions induced by a set of local states $\stateset_\wincon^\program \subset \stateset^\program$.
The corresponding set of configurations is $\confset_\wincon = \set{ \conf = \tuple{\statemap, \buffermap, \memorymap}_X \mid X \in \set{A,B} \land \exists \pid: \statemap(\pid) \in \stateset_\wincon^\program }$, that is, the set of all configurations where at least one process is in a state of $\stateset_\wincon^\program$.
The set $\confset_\wincon$ can then induce either a reachability or safety winning condition.
In the following, we will assume that the initial configuration (usually named $\conf_0$) is not contained in $\confset_\wincon$, since otherwise the game is decided immediately.
Furthermore, we desire that the process player immediately wins when reaching a target state in a reachability game, that is, we do not care whether the play can be extended infinitely or not.
Formally, we require that in a reachability game, a process cannot deadlock from a target state, implying that the process player cannot lose after reaching it.

\paragraph{Games on Single-Process Programs}
This section introduces games on single-process programs which will help us analysing the general case.
Given a game induced by a concurrent program $\program$, we compare it to the game on just one of the processes of $\program$.
We show that if the process player wins the single-process game, then she also wins the original game.
The main idea is that she achieves this by executing exactly the same instructions in both games.

For the remainder of this section, fix a program $\program = \tuple{\process^\pid}_{\pid \in \indexset}$ and a process index $\pid\in\indexset$.
Let $\program^\pid = \tuple{\process^\pid}$, i.e. the restriction of $\program$ to only the process $\process^\pid$.
Define $\game = \game(\program)$ and $\game^\pid = \game(\program^\pid)$, that is, the games induced by $\program$ and $\program^\pid$, respectively.
Let $\stateset_\wincon$ induce a reachability or safety winning condition for $\game$ and define the winning condition for $\game^\pid$ through $\stateset_\wincon^\pid = \stateset_\wincon \cap \stateset^\pid$.

Now, fix a configuration $\conf_0 \in \confset \setminus \confset_\wincon$ with empty buffers (i.e. $\buffermap(\conf_0) = \tuple\varepsilon_{\pid \in \indexset}$).
For $X \in \set{A,B}$ and a configuration $\conf = \tuple{ \statemap, \buffermap, \memorymap}_X \in \confset_X$, let $\conf\restrict = \tuple{\statemap\of\pid, \buffermap\of\pid, \memorymap}_X \in \confset_X^\pid$, which can be understood as the projection of $\conf$ onto the process $\process^\pid$.
Conversely, for a configuration $\conf^\pid \in \confset_X^\pid$ of $\game^\pid$,
define $\conf^\pid\extend = \tuple{ \statemap(\conf_0)[\pid \gets \statemap(\conf^\pid)], \buffermap(\conf_0)[\pid \gets \buffermap(\conf^\pid)], \memorymap(\conf^\pid) }_X \in \confset_X$,
that is, the configuration of $\game$ which is like $\conf^\pid$ for the process $\process^\pid$, but the local states and buffers of all other processes are as in the initial configuration $\conf_0$.
Note that for all $\conf^\pid$ of $\game^\pid$, it holds that $(\conf^\pid\extend)\restrict = \conf^\pid$.
On the other hand, $(\conf\restrict)\extend = \conf$ only holds for some $\conf$ of $\game$.

\begin{lemma}
\label{lem:extend}
    If the process player wins $\game^\pid$ starting from the configuration $\conf_0^\pid = \conf_0\restrict$, then she also wins $\game$ starting from $\conf_0$.
\end{lemma}
\begin{proof}
    Let $\sigma_A^\pid$ be a winning strategy of the process player in $\game^\pid$ and let $\sigma_A$ be the strategy in $\game$ defined by $\sigma_A(\conf) = \sigma_A^\pid(\conf\restrict)\extend$.
    Consider an arbitrary strategy $\sigma_B$ of the update player in $\game$ and let $\sigma_B^\pid$ be the strategy in $\game^\pid$ defined by $\sigma_B^\pid(\conf^\pid) = \sigma_B(\conf^\pid\extend)\restrict$.
    We write $\play = \play(\conf_0, \sigma_A, \sigma_B) := \conf_0, \conf_1, \dots$ and $\play^\pid = \play(\conf_0^\pid, \sigma_A^\pid, \sigma_B^\pid) := \conf_0^\pid, \conf_1^\pid, \dots$ and observe the following.

    First, $\conf_0^\pid\extend = \conf_0$.
    Next, assume $\conf_k^\pid\extend = \conf_k$ for all $k \leq n$ for some $n \in \Nat$.
    An immediate consequence is that $\conf_n\restrict = (\conf_n^\pid\extend)\restrict = \conf_n^\pid$.
    Now, if $\conf_n \in \confset_A$, then $\conf_{n+1} = \sigma_A(\conf_n) = \sigma_A^\pid(\conf_n\restrict)\extend = \sigma_A^\pid(\conf_n^\pid)\extend = \conf_{n+1}^\pid\extend$ by the assumption and the definition of $\sigma_A$.
    Otherwise, if $\conf_n \in \confset_B$, then $\conf_{n+1}^\pid = \sigma_B^\pid(\conf_n^\pid) = \sigma_B(\conf_n^\pid\extend)\restrict = \sigma_B(\conf_n)\restrict = \conf_{n+1}\restrict$.
    Since $\conf_n = \conf_n^\pid\extend$, the buffers of all processes except $\process^\pid$ are empty at $\conf_n$.
    Since the only action of the update player B is to update messages from the memory to the buffer, it follows that $\conf_n$ and $\conf_{n+1}$ must agree on the process states and buffers for the processes other than $\process^\pid$.
    Combined with the previously obtained chain of equalities, we can conclude that $\conf_{n+1}^\pid\extend = (\conf_{n+1}\restrict)\extend = \conf_{n+1}$.
    It follows by induction that $\conf_n^\pid\extend = \conf_n$ for all $n \in \Nat$.

    Since $\sigma_A^\pid$ is winning in $\game^\pid$, the play $\play^\pid = \play(\conf_0^\pid, \sigma_A^\pid, \sigma_B^\pid)$ is winning for player A.
    As $\conf_n^\pid\extend = \conf_n$ for all $n \in \Nat$, we can conclude that $\play^\pid$ and $\play = \play(\conf_0, \sigma_A, \sigma_B)$ visit the exactly same local states of $\process^\pid$.
    Note that the update player can never enter a deadlock, thus both plays must be infinite.
    Furthermore, since $\sigma_B$ was arbitrary, the strategy $\sigma_A$ must be winning for player A in $\game$, starting from $\conf_0$.

    In more detail:
    If $\game$ and $\game^\pid$ are reachability games, then the winning play $\play^\pid$ must visit at least one state in $\stateset_\wincon^\pid$.
    Since this is a subset of $\stateset_\wincon$ and the play $\play$ visits the same local states as $\play^\pid$, $\play$ must also be a winning play.
    On the other hand, if $\game$ and $\game^\pid$ are safety games, then neither $\play^\pid$ nor $\play$ visit any of the states in $\stateset_\wincon^\pid$.
    Since we required $\conf_0 \not\in \confset_\wincon$, the play $\play$ also does not visit any other local state of $\stateset_\wincon$ and must be winning.
\end{proof}

It is easy to see that the converse statement of \autoref{lem:extend} cannot hold for all $\pid \in \indexset$.
Rather, we only show that under certain conditions the process player is able to visit the same local states of a process $\process^1$ in both $\game$ and $\game^\pid$.
The strategy to do so will take a specific play in $\game$ and mimic all instructions that have been played in $\process^\pid$, similar as in the previous proof.

Fix $\pid \in \indexset$.
Let $\sigma_A$ be a winning strategy for the process player and let $\sigma_B$ be the strategy of the update player where she never updates any buffer messages to the memory.
Consider the play $\play(\conf_0, \sigma_A, \sigma_B)$ in $\game$.
For $k = 1, 2, \dots$, let $\bar\conf_k \to \conf_k$ be the $k$-th time in this play where either the local state or the buffer of $\process^\pid$ changes.
This transition is due to some unique instruction $\statemap(\bar\conf_k)(\pid) \to[\instr_k] \statemap(\conf_k)(\pid)$ in $\process^\pid$.
In particular, it cannot be due to a memory update, since $\sigma_B$ was chosen that way.
Note that there does not necessarily need to be an infinite amount of $k$ with this property.
We define the strategy $\sigma^\pid_A$ as follows:
Whenever $\game^\pid$ is in the $k$-th round, the local state of the process is $\statemap(\bar\conf_k)(\pid)$ and $\instr_k$ is enabled, execute this instruction to move to the unique configuration with local state $\statemap(\conf_k)(\pid)$.
Otherwise, make an arbitrary move.
Let $\sigma^\pid_B$ be an arbitrary strategy of the update player for $\game^\pid$.
After the k-th round of $\play(\conf^\pid_0, \sigma^\pid_A, \sigma^\pid_B)$, the game is in some position $\conf^\pid_k$.

\begin{claim}
\label{claim:single-process}
    For all $k \in \Nat$ for which $\conf_k$ is defined, it holds that $\conf_k\restrict \to[\up\kstar] \conf^\pid_k$.
\end{claim}
\begin{proof}
    We prove the claim by induction.
    For the base case $k = 0$ define $\conf^\pid_0 = \conf_0\restrict$.
    In the induction step assume $k > 0$ and consider the game before the $k$-th round, where it is in configuration $\conf^\pid_{k-1}$.
    We claim that $\instr_k$ is enabled at $\conf^\pid_{k-1}$.
    First, note that $\bar\conf_k\restrict = \conf_{k-1}\restrict$ which follows directly from their definitions.
    Then, consider the possible instructions $\instr$.
    If $\instr = \mf$, then this means that the buffer of $\process^\pid$ at $\bar\conf_k$ is empty.
    Together with the induction hypothesis it follows that $\bar\conf_k\restrict = \conf_{k-1}\restrict = \conf^\pid_{k-1}$ and $\mf$ is enabled.
    In the case that $\instr = \rd\of\xd$, it holds for $\bar\conf_k$ that either $\dval$ is the initial value of $\xvar$ and $\process^\pid$ has not written to it yet, or the last write of $\process^\pid$ to $\xvar$ was $\dval$.
    We can easily see that this still holds true after performing any amount of update operations.
    Thus, $\rd\of\xd$ is also enabled at $\conf^\pid_{k-1}$, since $\bar\conf_k\restrict = \conf_{k-1}\restrict \to[\up\kstar] \conf^\pid_{k-1}$.
    In all other cases, there is nothing to show, since the instruction is always enabled.

    By the construction of $\sigma^\pid_A$, the process player executes $\instr_k$ and the game moves from $\conf^\pid_{k-1}$ to a configuration $\tilde\conf^\pid_k$ with $\statemap(\tilde\conf^\pid_k) = \statemap(\conf_k)(\pid)$.
    If $\instr_k$ is a write instruction, then the same message is added to the buffers of both $\game$ and $\game^\pid$.
    Otherwise, the buffers stay the same.
    We conclude that $\conf_k\restrict \to[\up\kstar] \tilde\conf^\pid_k$.
    Next, the update player moves from $\tilde\conf^\pid_k$ to $\conf^\pid_k$, performing some amount of update operations.
    Thus, $\conf_k\restrict \to[\up\kstar] \tilde\conf^\pid_k \to[\up\kstar] \conf^\pid_k$, which concludes the proof by induction.
\end{proof}

\paragraph{Concurrent Games}
We combine the results for single-process games to obtain the following theorem.
\begin{theorem}
\label{thm:reduction}
    The process player wins $\game$ starting from a configuration $\conf_0$ if and only if she also wins $\game^\pid$ starting from configuration $\conf_0\restrict$ for at least one $\pid \in \indexset$.
\end{theorem}
\begin{proof}
    By \autoref{lem:extend}, if the process player wins $\game^\pid$ for at least one $\pid \in \indexset$, then she also wins $\game$.
    For the other direction, consider the strategies as defined above.
    What is left to show is that $\sigma^\pid_A$ is winning for at least one $\pid \in \indexset$.

    If the process player has a reachability objective, $\play(\conf_0, \sigma_A, \sigma_B)$ visits at least one target state in some process $\process^\pid$.
    Note again that the update player cannot be deadlocked and therefore the play must be infinite.
    From \autoref{claim:single-process} it follows that $\play(\conf^\pid_0, \sigma^\pid_A, \sigma^\pid_B)$ visits the same local states of $\process^\pid$ than $\play(\conf_0, \sigma_A, \sigma_B)$ and in particular, it visits the same target state.
    Since $\sigma^\pid_B$ was chosen arbitrarily, it means that $\sigma^\pid_A$ is a winning strategy.
    Otherwise, if the process player has a safety objective, $\play(\conf_0, \sigma_A, \sigma_B)$ executes an infinite amount of instructions in at least one process $\process^\pid$, but never visits any of its target states.
    Using the same arguments as previously, it follows that $\play(\conf^\pid_0, \sigma^\pid_A, \sigma^\pid_B)$ is also a winning play and $\sigma^\pid_A$ is a winning strategy.
\end{proof}

\autoref{thm:reduction} reduces the reachability problem for games on concurrent programs to the single-process case.
Although the game $\game^\pid$ still has an infinite amount of configurations, many of them are indistinguishable in the sense that they have the same local state and allow the same sequences of instructions to be executed.
\autoref{sec:complexity} formally constructs a finite game that is a so-called \emph{bisimulation} of $\game^\pid$.
Furthermore, it shows that the reachability and safety problems for TSO games are \pspace-complete.

\paragraph{Fairness Conditions}
In the previous section we have seen that the game played under TSO semantics reduces to the analysis of games on single-process programs.
This is somewhat unsatisfying, since those games do not exhibit any concurrent behaviour that arises from the communication between multiple processes.
The underlying reason is the structure of the optimal strategies of the two players:
If the process player wins, it is because she only plays in one single process.
Otherwise, the update player wins by never updating any buffer messages.
Both of these behaviours are not natural in the sense that they will not occur in any reasonable program environment:
We should be able to assume that eventually (1) every buffer message will be updated to the memory and (2) every process will execute an instruction.
In \autoref{sec:update-fairness} and \autoref{sec:process-fairness} we will impose additional restrictions on the two players to enforce this behaviour.

\section{Complexity of Single-Process Games}
\label{sec:complexity}

\paragraph{TSO Views}
In the single-process program $\program^\pid$, there is no communication between different processes.
A read operation of the process $\process^\pid$ on a variable $\xvar$ either reads the initial value from the (shared) memory, or the value of the last write on $\xvar$ done by $\process^\pid$, if such a write operation has happened.
In the latter case, the value of the read operation can come from either the buffer of $\process^\pid$ or directly from the memory.
But a single process cannot distinguish between these two cases.
To be exact, the information that the process can obtain from the buffer and the memory is the value that $\process^\pid$ can read from each variable, and whether the process can execute a memory fence instruction or not.
Together with the local state of $\process^\pid$ at the current configuration, this completely determines the enabled transitions for the process.

We call this concept the \emph{view} of the process on the (concurrent) system and define it formally as a tuple $\view = \tuple{ \statemap, \valuemap, \fencemap }$, where:
\begin{itemize}
    \item $\statemap \in \stateset^\pid$ is the local state of $\process^\pid$.
    \item $\valuemap: \varset \to \valset$ defines which value $\process^\pid$ reads from each variable.
    \item $\fencemap \in \set{ \true, \false }$ represents the possibility to perform a memory fence instruction.
\end{itemize}
Given a view $\view = \tuple{ \statemap, \valuemap, \fencemap }$, we write $\statemap\of\view$, $\valuemap\of\view$ and $\fencemap\of\view$ for the local state $\statemap$, the value state $\valuemap$ and the fence state $\fencemap$ of $\view$, respectively.
The view of a configuration $\conf \in \confset^\pid$ is denoted by $\view\of\conf$ and defined in the following way.
First, $\statemap(\view\of\conf) = \statemap(\conf)$.
For all $\xvar \in \varset$, if $\buffermap(\conf)|_{\set\xvar\times\valset} = \tuple\xd \cdot \word$, then $\valuemap(\view\of\conf)(\xvar) = \dval$.
Otherwise, $\valuemap(\view\of\conf)(\xvar) = \memorymap(\conf)\of\xvar$.
Lastly, $\fencemap(\view\of\conf) = \true$ if and only if $\buffermap(\conf) = \varepsilon$.

For $\conf_1, \conf_2 \in \confset$, if $\view(\conf_1) = \view(\conf_2)$, then we write $\conf_1 \equiv \conf_2$.
In such a case, the process $\process^\pid$ cannot differentiate between $\conf_1$ and $\conf_2$ in the sense that the enabled transitions in both configurations are the same.
This is shown in \autoref{fig:views} and formally captured in \autoref{lem:views}.

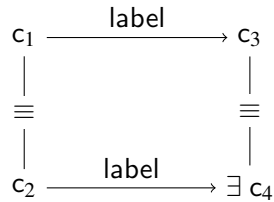
\begin{figure}[h]
\centering
\begin{tikzpicture}[xscale=3, yscale=-2]
    \node at (0,0) (c1) {$\conf_1$};
    \node at (1,0) (c2) {$\conf_3$};
    \node at (0,1) (c3) {$\conf_2$};
    \node at (1,1) (c4) {$\exists\ \conf_4$};

    \draw[->] (c1) -- node[above] {$\lbl$} (c2);
    \draw[->] (c3) -- node[above] {$\lbl$} (c4);
    \draw[- ] (c1) -- node[fill=white] {$\equiv$} (c3);
    \draw[- ] (c2) -- node[fill=white] {$\equiv$} (c4);
\end{tikzpicture}
\caption{The configurations of \autoref{lem:views}.}
\label{fig:views}
\end{figure}

\begin{lemma}
\label{lem:views}
    For all $\conf_1, \conf_3, \conf_2 \in \confset^\pid$ and $\lbl \in \instrs \cup \set{ \up\kstar }$ with $\conf_1 \equiv \conf_2$ and $\conf_1 \to[\lbl] \conf_3$,
    there exists a $\conf_4 \in \confset^\pid$ such that $\conf_3 \equiv \conf_4$ and $\conf_2 \to[\lbl] \conf_4$.
\end{lemma}
\begin{proof}
    If $\lbl = \up\kstar$, this clearly holds for $\conf_4 = \conf_2$.
    Otherwise, we first show that $\lbl \in \instrs$ is enabled at $\conf_2$.
    Since $\conf_1 \equiv \conf_2$, it holds that $\statemap(\conf_1) = \statemap(\conf_2)$.
    Furthermore, if $\lbl = \rd\of\xd$, then $\valuemap(\view(\conf_1))(\xvar) = \valuemap(\view(\conf_2))(\xvar) = \dval$.
    Also, if $\lbl = \mf$, then $\buffermap(\view(\conf_1)) = \varepsilon$ and since $\fencemap(\view(\conf_1)) = \fencemap(\view(\conf_2)) = \true$ it follows that $\buffermap(\view(\conf_2)) = \varepsilon$.
    From these considerations and the definition of the TSO semantics (see \autoref{fig:tso-semantics}), it follows that $\lbl$ is indeed enabled at $\conf_2$.

    Let $\conf_4$ be the unique configuration obtained after executing the transition $\statemap(\conf_1) \to[\lbl] \statemap(\conf_3)$ at $\conf_2$, i.e. $\conf_2 \to[\lbl] \conf_4$ and $\statemap(\conf_4) = \statemap(\conf_3)$.
    If $\lbl = \wr\of\xd$, then $\valuemap(\view(\conf_4)) = \valuemap(\view(\conf_3)) = \valuemap(\view(\conf_1))[\xvar \gets \dval]$
    and $\fencemap(\view(\conf_4)) = \fencemap(\view(\conf_3)) = \false$.
    Otherwise, $\valuemap(\view(\conf_4)) = \valuemap(\view(\conf_3)) = \valuemap(\view(\conf_1))$
    and $\fencemap(\view(\conf_4)) = \fencemap(\view(\conf_3)) = \fencemap(\view(\conf_1))$.
    In all cases it follows that $\conf_3 \equiv \conf_4$.
\end{proof}

\paragraph{Bisimulations}
A \emph{colouring} of a game $\game$ is a function $\lambda: \confset \to \colset$ from the set of configurations into some set of colours $\colset$.
Consider two games $\game$ and $\game'$ with colouring functions $\lambda$ and $\lambda'$, respectively.
A \emph{bisimulation} (also called \emph{zig-zag relation}) between $\game$ and $\game'$ is a relation $\bisim \subseteq \confset \times \confset'$ such that for all pair of related configurations $(\conf_1, \conf_2) \in \bisim$ it holds that (cf. \autoref{fig:bisim}):
\begin{itemize}
    \item $\conf_1$ and $\conf_2$ agree on their colour: $\lambda(\conf_1) = \lambda'(\conf_2)$
    \item (\emph{zig}) for each transition $\conf_1 \to[\lbl] \conf_3$ there is a transition $\conf_2 \to[\lbl] \conf_4$ such that $(\conf_3, \conf_4) \in \bisim$.
    \item (\emph{zag}) for each transition $\conf_2 \to[\lbl] \conf_4$ there is a transition $\conf_1 \to[\lbl] \conf_3$ such that $(\conf_3, \conf_4) \in \bisim$.
\end{itemize}
We say that two related configurations $\conf_1$ and $\conf_2$ are \emph{bisimilar} and write $\conf \approx \conf'$.
We call $\game$ and $\game'$ \emph{bisimilar} if there is a bisimulation between them.

\begin{figure}
\centering
\begin{subfigure}[b]{0.4\linewidth}
\centering
\begin{tikzpicture}[xscale=3, yscale=-2]
    \node at (0,0) (c1) {$\conf_1$};
    \node at (1,0) (c2) {$\conf_3$};
    \node at (0,1) (c3) {$\conf_2$};
    \node at (1,1) (c4) {$\exists\ \conf_4$};

    \draw[->] (c1) -- node[above] {$\lbl$} (c2);
    \draw[->] (c3) -- node[above] {$\lbl$} (c4);
    \draw[- ] (c1) -- node[fill=white] {$\approx$} (c3);
    \draw[- ] (c2) -- node[fill=white] {$\approx$} (c4);
\end{tikzpicture}
\caption{zig property}
\end{subfigure}
\begin{subfigure}[b]{0.4\linewidth}
\centering
\begin{tikzpicture}[xscale=3, yscale=-2]
    \node at (0,0) (c1) {$\conf_1$};
    \node at (1,0) (c2) {$\exists\ \conf_3$};
    \node at (0,1) (c3) {$\conf_2$};
    \node at (1,1) (c4) {$\conf_4$};

    \draw[->] (c1) -- node[above] {$\lbl$} (c2);
    \draw[->] (c3) -- node[above] {$\lbl$} (c4);
    \draw[- ] (c1) -- node[fill=white] {$\approx$} (c3);
    \draw[- ] (c2) -- node[fill=white] {$\approx$} (c4);
\end{tikzpicture}
\caption{zag property}
\end{subfigure}
\caption{Configurations in a bisimulation.}
\label{fig:bisim}
\end{figure}
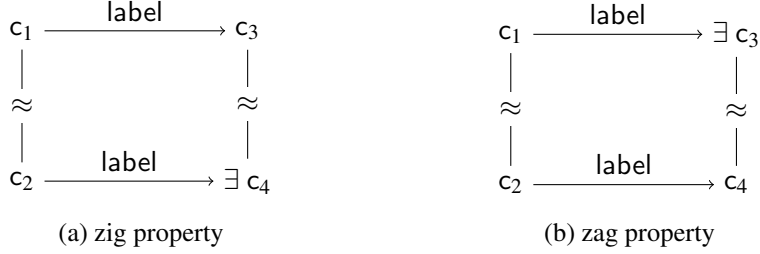

It is common knowledge in game theory that winning strategies are preserved under bisimulations if the colourings are a refinement of the winning condition in the following sense.
Consider two bisimilar games $\game$ and $\game'$ with winning conditions given by $\confset_\wincon$ and $\confset'_\wincon$, respectively.
Let $\lambda$ be a colouring function for $\game$ such that the configurations in $\confset_\wincon$ have different colours than the rest of the configurations, i.e. $\lambda^{-1}(\lambda(\confset_\wincon)) = \confset_\wincon$.
Define $\lambda'$ as a colouring function for $\game'$ accordingly.

\begin{lemma}
    Given two bisimilar configurations $\conf_0 \in \confset$ and $\conf_0' \in \confset'$, it holds that $\conf_0$ is a winning configuration in $\game$ if and only if $\conf_0'$ is a winning configuration in $\game'$.
\end{lemma}

We define a game on views $\game^\viewset = \tuple{ \viewset, \viewset_A, \viewset_B, \to_\view}$ that is bisimilar to the single-process game $\game^\pid$.
Let $\viewset_X = \set{ \view(\conf)_X \mid \conf \in \confset^\pid }$ for $X \in \set{A,B}$ and $\viewset = \viewset_A \cup \viewset_B$.
We extend the notation of the function $\view$ to game configurations by $\view(\conf_X) = \view(\conf)_X$ for $X \in \set{A,B}$ and $\conf \in \confset^\pid$.
Now, we can define $\to_\view$ by $\view(\conf) \to[\lbl]_\view \view(\conf')$ if and only if $\conf \to[\lbl] \conf'$ for some $\conf, \conf' \in \confset^\pid$.

\begin{theorem}
    The relation $\bisim = \set{ (\conf, \view\of\conf) \mid \conf \in \confset^\pid} \subset \confset^\pid \times \confset^\viewset$
    is a bisimulation between $\game^\pid$ and $\game^\viewset$
    with colouring functions $\lambda^\pid: \confset^\pid \to \viewset, \conf \mapsto \view(\conf)$ and $\lambda^\viewset = \operatorname{id}_\viewset$, respectively.
\end{theorem}
\begin{proof}
    From the definition it follows directly that related configurations agree on their colour and that $\game^\viewset$ can simulate $\game^\pid$.
    What is left to show is that $\game^\pid$ can also simulate $\game^\viewset$, i.e. that for all $\conf \approx \view$ and $\view \to[\lbl]_\view \tilde\view$ there is $\conf \to[\lbl] \tilde\conf$ with $\tilde\conf \approx \tilde\view$.
    The transition $\view \to[\lbl]_\view \tilde\view$ is due to a transition $\altconf \to[\lbl] \tilde\altconf$ for some $\altconf, \tilde\altconf \in \confset^\pid$ with $\altconf \approx \view$ and $\tilde\altconf \approx \tilde\view$.
    Since $\view\of\conf = \view = \view\of\altconf$, it follows that $\conf \equiv \altconf$.
    Apply \autoref{lem:views} to $\conf_1 = \altconf$, $\conf_2 = \conf$ and $\conf_3 = \tilde\altconf$ to obtain a configuration $\tilde\conf = \conf_4$ with the desired properties.
\end{proof}

Since $\confset^\viewset$ is finite, it is rather evident that the reachability and safety problem are decidable, e.g. by applying a backward induction algorithm.
In fact, both problems are $\pspace$-complete.
Intuitively, this makes sense since each variable can be seen as a single cell of a bounded Turing machine.
In the following, we will give a polynomial-space algorithm to show the upper complexity bound and construct a reduction from TQBF for the lower bound.
These results then immediately translate to the single-process TSO game $\game^\pid$, since it is bisimilar to $\game^\viewset$.

\paragraph{\pspace-Membership}
We first examine the possible moves of the update player.
If the program is in some configuration $\view$ where $\fencemap\of\view = \true$, then there is only one valid move which is to stay in the same configuration.
Otherwise, if $\fencemap\of\view = \false$, then the update player can choose to change the value to $\true$, which corresponds to updating all messages from the buffer to the memory, thus enabling memory fence instructions.

Similar to what we have seen in \autoref{sec:tso-games}, we can argue that the update player can follow a very simple strategy:
Since changing the value of $\fencemap\of\view$ from $\false$ to $\true$ never disables any transitions, but might enable some transitions, it is never beneficial for the update player to do so.
More formally, consider a strategy $\sigma_A$ of the process player that wins against the strategy $\sigma_B$ of the update player where she never updates.
We define another strategy by $\sigma'_A(\view) = \sigma_A(\tuple{\statemap\of\view, \valuemap\of\view, \false})$.
This is well-defined since all outgoing transitions at $\tuple{\statemap\of\view, \valuemap\of\view, \false}$ are also enabled at $\view$.
It follows that $\sigma'_A$ wins against any strategy $\sigma'_B$ since $\play(\view_0, \sigma_A, \sigma_B)$ and $\play(\view_0, \sigma'_A, \sigma'_B)$ visit the same local states for any initial view $\view$.
In the following, we will assume that the update player never changes $\fencemap$.
This means that she is a dummy player that only has one available transition at all times.

To solve the reachability problem, we nondeterministically simulate a play step by step.
We store the current configuration along the way, but disregard the history.
Instead, we keep a counter on how many steps we have simulated.
If the simulation runs for more than $\sizeof\viewset$ steps without visiting the set of target states $\viewset_\wincon$, it must have entered a loop and can terminate.
Since
$$ \log(\sizeof\viewset) = \log(\sizeof\stateset \cdot \sizeof\valset^{\sizeof\varset} \cdot 2) = \log(\sizeof\stateset) + \sizeof\varset \cdot \log(\sizeof\valset) + 1 \quad $$
is polynomial in the size of $\game^\viewset$, we can store the counter in polynomial space.
The first time the simulation reaches $\viewset_\wincon$, the counter is reset.
If the simulation then continues for more than $\sizeof\viewset$ steps without entering a deadlock, it entered a loop which can be extended to an infinite winning play.

For the safety problem, we only need to perform the second part of the simulation where we search for a loop.
Of course, this time we need to check that the loop does not visit $\viewset_\wincon$.

We described a nondeterministic polynomial-space algorithm to solve the reachability and safety problem.
Since $\npspace = \pspace$ by Savitch's theorem, we can conclude that both problems are in $\pspace$.

\paragraph{\pspace-Hardness}
\begin{figure}
\centering

\tikzset{io/.style={draw, circle, fill=white, minimum size=21pt}}

\begin{subfigure}[b]{0.9\linewidth}
\centering
\begin{tikzpicture}[yscale=-0.9]
\draw[draw=black] (0,0) rectangle (10,3.5);
    \node[io] at (2,0) (I) {I};
    \node[io] at (8,0) (S) {O};

    \draw[draw=black] (0.5,2) rectangle (9.5,3);
    \node[io] at (2,2) (i) {I};
    \node[io] at (8,2) (s) {O};
    \node at (5,2.5) {$Q_{k+1}\ x_{k+1} \dots$};

    \draw[->] (I) to[bend left ] node[left]  {$\wr(\xvar_k, 0)$} (i);
    \draw[->] (I) to[bend right] node[right] {$\wr(\xvar_k, 1)$} (i);
    \draw[->] (s) -- node[right] {$\nop$} (S);
\end{tikzpicture}
\caption{Gadget for $\exists\ x_k\ Q_{k+1}\ x_{k+1} \dots$.}
\end{subfigure}

\medskip

\begin{subfigure}[b]{0.9\linewidth}
\centering
\begin{tikzpicture}[yscale=-0.9]
    \draw[draw=black] (0,0) rectangle (10,3.5);
    \node[io] at (2,0) (I) {I};
    \node[io] at (8,0) (S) {O};
    \node[io] at (5,1) (h) { };

    \draw[draw=black] (0.5,2) rectangle (9.5,3);
    \node[io] at (2,2) (i) {I};
    \node[io] at (8,2) (s) {O};
    \node at (5,2.5) {$Q_{k+1}\ x_{k+1} \dots$};

    \draw[->] (I) -- node[left]  {$\wr(\xvar_k, 0)$} (i);
    \draw[->] (s) -- node[right] {$\rd(\xvar_k, 1)$} (S);
    \draw[->] (s) to[out=240, in=  0] node[above] {$\rd(\xvar_k, 0)$} (h);
    \draw[->] (h) to[out=180, in=-60] node[above] {$\wr(\xvar_k, 1)$} (i);
\end{tikzpicture}
\caption{Gadget for $\forall\ x_k\ Q_{k+1}\ x_{k+1} \dots$.}
\end{subfigure}

\medskip

\begin{subfigure}[b]{0.9\linewidth}
\centering
\begin{tikzpicture}[yscale=-0.9]
    \draw[draw=black] (0,0) rectangle (10,3.5);
    \node[io] at (2,0) (I) {I};
    \node[io] at (8,0) (S) {O};

    \draw[draw=black] (0.5,2) rectangle (4.5,3);
    \node[io] at (2,2) (i1) {I};
    \node[io] at (3,2) (s1) {O};
    \node at (2.5,2.5) {$\phi_1$};

    \draw[draw=black] (5.5,2) rectangle (9.5,3);
    \node[io] at (7,2) (i2) {I};
    \node[io] at (8,2) (s2) {O};
    \node at (7.5,2.5) {$\phi_2$};

    \draw[->] (I)  -- node[left]        {$\nop$} (i1);
    \draw[->] (I)  -- node[below left]  {$\nop$} (i2);
    \draw[->] (s1) -- node[below right] {$\nop$} (S);
    \draw[->] (s2) -- node[right]       {$\nop$} (S);
\end{tikzpicture}
\caption{Gadget for $\phi = \phi_1 \lor \phi_2$.}
\end{subfigure}

\medskip

\begin{subfigure}[b]{0.9\linewidth}
\centering
\begin{tikzpicture}[yscale=-0.9]
    \draw[draw=black] (0,0) rectangle (10,3.5);
    \node[io] at (2,0) (I) {I};
    \node[io] at (8,0) (S) {O};

    \draw[draw=black] (0.5,2) rectangle (4.5,3);
    \node[io] at (2,2) (i1) {I};
    \node[io] at (3,2) (s1) {O};
    \node at (2.5,2.5) {$\phi_1$};

    \draw[draw=black] (5.5,2) rectangle (9.5,3);
    \node[io] at (7,2) (i2) {I};
    \node[io] at (8,2) (s2) {O};
    \node at (7.5,2.5) {$\phi_2$};

    \draw[->] (I)  -- node[left]        {$\nop$} (i1);
    \draw[->] (s1) to[bend right] node[above] {$\nop$} (i2);
    \draw[->] (s2) -- node[right]       {$\nop$} (S);
\end{tikzpicture}
\caption{Gadget for $\phi = \phi_1 \land \phi_2$.}
\end{subfigure}

\medskip

\begin{subfigure}[b]{0.45\linewidth}
\centering
\begin{tikzpicture}[yscale=-0.9]
    \draw[draw=black] (0,0) rectangle (5,1.5);
    \node[io] at (1,0) (I) {I};
    \node[io] at (4,0) (S) {O};
    \draw[->] (I) to[bend left] node[below] {$\xrd(\xvar_k, 1)$} (S);
\end{tikzpicture}
\caption{Gadget for $\phi = x_k$.}
\end{subfigure}
\hfill
\begin{subfigure}[b]{0.45\linewidth}
\centering
\begin{tikzpicture}[yscale=-0.9]
    \draw[draw=black] (0,0) rectangle (5,1.5);
    \node[io] at (1,0) (I) {I};
    \node[io] at (4,0) (S) {O};
    \draw[->] (I) to[bend left] node[below] {$\xrd(\xvar_k, 0)$} (S);
\end{tikzpicture}
\caption{Gadget for $\phi = \neg x_k$.}
\end{subfigure}

\caption{Reduction from TQBF to TSO games.}
\label{fig:tqbf}
\end{figure}
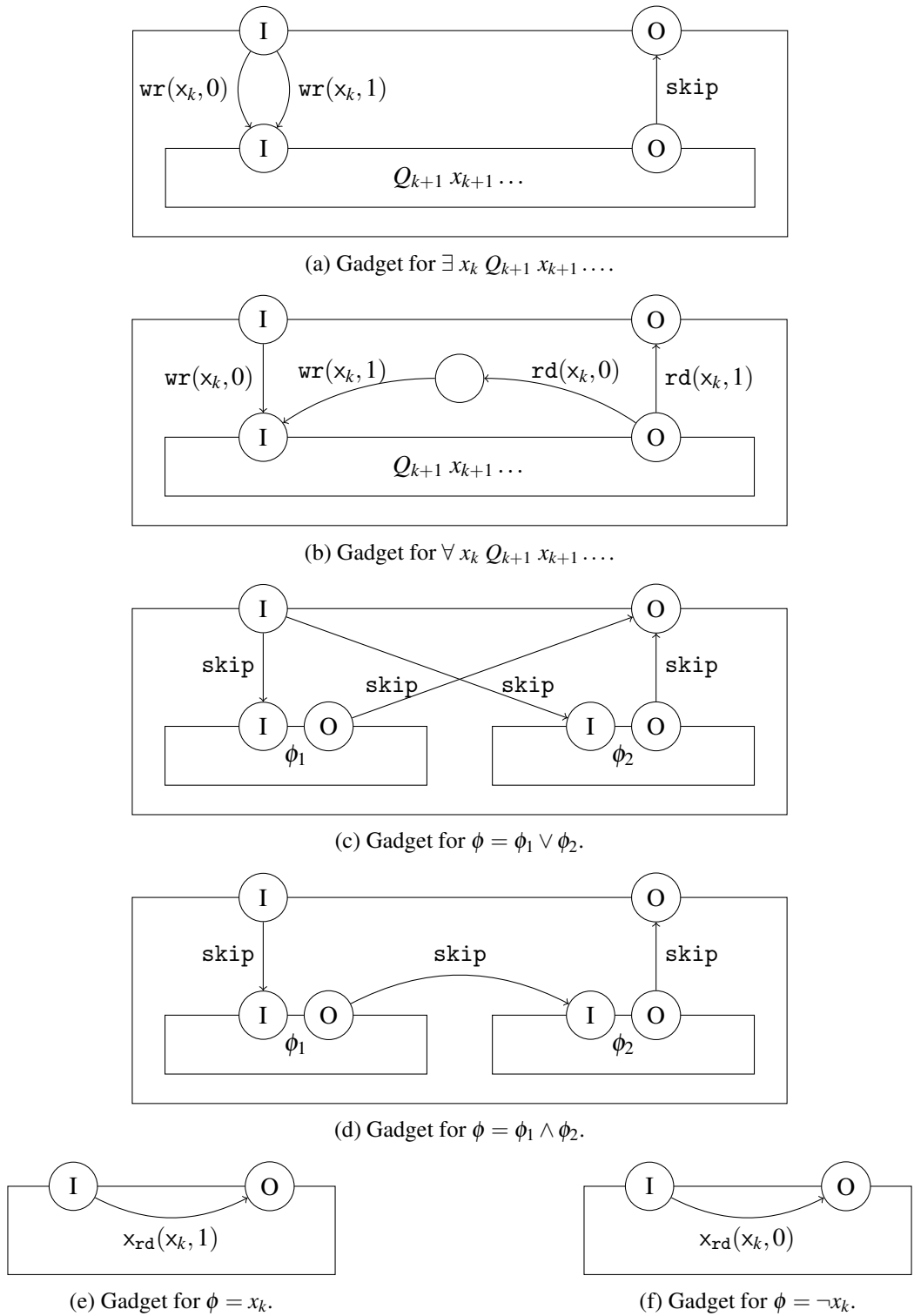

Given a QBF formula $\varphi = Q_1\ x_1\ \dots\ Q_n\ x_n\ \phi(x_1, \dots, x_n)$ where $Q_k \in \set{\exists, \forall}$, we construct a single-process TSO program that can decide whether $\varphi$ is true.
More precisely, the process player wins the view game induced by the program if and only if the QBF formula is true.
The winning condition can be specified both as a reachability or safety objective.

The program consists of one process and $n$ variables $\xvar_1, \dots, \xvar_n$.
For each quantifier $Q_1, \dots, Q_n$ of $\varphi$ and each subformula of $\varphi$, we create a gadget of states and transitions (\autoref{fig:tqbf}).
Every gadget has an initial state (I) and an output state (O) which is reachable by the process player if and only if the corresponding subformula is true given the current (partial) variable assignment.

If $Q_k = \exists$, then the initial state of the gadget for $Q_k$ has two outgoing transitions with labels $\wr(\xvar_k, 0)$ and $\wr(\xvar_k, 1)$.
Both lead to the initial state of the gadget for $Q_{k+1}$, or the gadget of $\phi$ if $k=n$.
Additionally, there is a $\nop$ transition from the output of the successor gadget to the output of the gadget for $Q_k$.
The gadget for $Q_k = \forall$ is similar, but there is only one transition from the initial state to the successor gadget, labelled with $\wr(\xvar_k, 0)$.
After the output gadget of the successor is reached, the gadget for $Q_k$ checks the value of $\xvar_k$.
If it is $0$, the variable is updated to $1$ and the successor gadget is entered again.
Otherwise, the transition leads to the output state.

The gadgets for the subformulas of $\phi$ are as follows.
For $\phi_1 \lor \phi_2$, two $\nop$ transition lead to the inputs of the gadgets for $\phi_1$ and $\phi_2$, respectively.
The output state of both gadgets are connected to the output state of the gadget of $\phi_1 \lor \phi_2$.
The gadget of $\phi_1 \land \phi_2$ contains three $\nop$ transitions:
From its initial state to the input of $\phi_1$, from the output of $\phi_1$ to the input of $\phi_2$ and from the output of $\phi_1$ to the output of $\phi_1 \land \phi_2$.
The gadgets for $x_k$ and $\neg x_k$ only consist of only one transition each, going from their initial state to their output state.
In the first case, the label of the transition is $\rd(\xvar_k, 1)$, in the latter case it is $\rd(\xvar_k, 0)$.

We outline how the game will be played by the process player.
At an existential quantifier, she chooses a value for $x_k$ that satisfies the formula and stores her choice in $\xvar_k$.
At a universal quantifier, she does not have a choice but needs to show that the formula is satisfied with both values for $\xvar_k$, one after the other.
Similarly, in the boolean formula, if she encounters a disjunction, she can choose either subformula, but if she encounters a conjunction, she has to satisfy both of them.
Whenever the player reaches a literal, she checks if the correct value is assigned to the corresponding variable.

By an inductive approach it is straightforward to verify that the program can reach the output state of a gadget if and only if the formula is satisfiable with respect to the current (partial) variable assignment.
We conclude that starting from the initial state of the gadget for $Q_1$, the process player reaches the output state of $Q_1$ if and only if $\varphi$ is true.
Otherwise, the player will reach a deadlock in one of the literal gadgets.

To allow for infinite plays, we add a self-loop to the output state of the first gadget.
We define the winning conditions in a very simple way:
The output state of $Q_1$ induces a reachability condition, while the empty set represents a safety condition.
In both cases, the process player wins if and only if she reaches said state, since it is the only way to enter an infinite play.

Since TQBF is $\pspace$-hard, it follows that both the reachability and safety problem of TSO view games is also $\pspace$-hard.

\section{Update Fairness in Reachability Games}
\label{sec:update-fairness}

In this section, we introduce \emph{update fairness}, which we require the update player to satisfy.
The core idea of update fairness is that eventually, each buffer message will be updated to the shared memory.
This means that the process player can in some sense \emph{wait} for the buffer messages to arrive in the memory.
In safety games, delaying the run indefinitely favours the process player.
Thus, we will focus on reachability games in this section.

We implement update fairness as follows.
Whenever the program is in a configuration at which no program instruction is enabled (a deadlock), the system waits for the update player to update buffer messages to the memory, until the program exits the deadlock (or all buffers are empty).
To simplify the formalisation, we will assume that the update player does so in her very next turn.
This idea can be equivalently expressed by saying that if it is the process player's turn and the system is deadlocked, then it follows that all buffers must be empty.

Let $\program = \tuple{\process^\pid}_{\pid\in\indexset}$ be a TSO program with induced game $\game\of\program$.
We define $\wincon_U$ as the set of all plays $\play = \conf_0, \conf_1, \dots$ in $\game\of\program$ that satisfy update fairness:
$$ \play \in \wincon_U \iff \forall\ k \in \Nat, \conf_k \in \confset_A: \left(
\post(\conf_k) = \emptyset \implies \forall\ \pid\in\indexset: \buffermap(\conf_k)(\pid) = \varepsilon
\right) $$
For a reachability condition $\wincon_R$, let the set $\wincon_{RU} = \wincon_R \cup \overline{\wincon_U} = \set{\play \mid \play \in \wincon_R \lor \play \not\in \wincon_U}$ be the set of winning plays for the process player, i.e. the set of all plays that either reach a target state or that do not admit update fairness.
The remainder of this section will be dedicated to show that the reachability problem under update fairness is undecidable.
We will achieve this by reducing the state reachability problem of perfect channel systems, which is undecidable, to the reachability problem of $\game\of\program$ with respect to $\wincon_{RU}$.
The main ideas of the reduction are similar to those in \cite{DBLP:journals/corr/abs-2310-00990}.

Given a perfect channel system $\channelsystem = \tuple{ \channelstateset, \channelmessageset, \transition }$, we construct a TSO program $\program$ that simulates $\channelsystem$.
The process player will decide which transitions of the PCS to simulate, while the update player only takes care of the buffer updates.
The program consists of two processes $\process^1$ and $\process^2$, which are shown in \autoref{fig:ru-reduction-1} and \autoref{fig:ru-reduction-2}, respectively.

\begin{figure}[tbh]
\centering
\begin{minipage}[b]{.45\textwidth}
    \centering

\begin{subfigure}{\linewidth}
\centering
\begin{tikzpicture}[yscale=-1]
    \node at (0,0) (q1) {$\channelstate$};
    \node at (0,1) (q2) {$\channelstate'$};

    \draw[->] (q1) -- node[right] {$\nop$} (q2);
\end{tikzpicture}
\caption{skip operation $\channelstate \to[\nop]_\channelsystem \channelstate'$}
\end{subfigure}
\bigskip
\bigskip

\begin{subfigure}{\linewidth}
\centering
\begin{tikzpicture}[yscale=-1]
    \node at (0,0) (q1) {$\channelstate$};
    \node at (0,1) (h1) {$\hstate_1$};
    \node at (0,2) (q2) {$\channelstate'$};

    \draw[->] (q1) -- node[right] {$\wr(\xwr,\channelmessage)$} (h1);
    \draw[->] (h1) -- node[right] {$\wr(\yvar,1)$} (q2);
\end{tikzpicture}
\caption{send operation $\channelstate \to[!\channelmessage]_\channelsystem \channelstate'$}
\end{subfigure}
\bigskip
\bigskip

\begin{subfigure}{\linewidth}
\centering
\begin{tikzpicture}[yscale=-1]
    \node at (0,0) (q1) {$\channelstate$};
    \node at (0,1) (h1) {$\hstate_1$};
    \node at (0,2) (h2) {$\hstate_2$};
    \node at (0,3) (q2) {$\channelstate'$};

    \draw[->] (q1) -- node[right] {$\nop$} (h1);
    \draw[->] (h1) -- node[right] {$\rd(\xrd,\channelmessage)$} (h2);
    \draw[->] (h2) -- node[right] {$\rd(\xrd,\bot)$} (q2);
\end{tikzpicture}
\caption{receive operation $\channelstate \to[?\channelmessage]_\channelsystem \channelstate'$}
\label{fig:ru-reduction-1c}
\end{subfigure}
\caption{$\process^1$ of the reduction from PCS.}
\label{fig:ru-reduction-1}
\end{minipage}%
\hfill
\begin{minipage}[b]{.5\textwidth}
    \centering
\begin{tikzpicture}[yscale=-1]
    \node at (0,1) (q1) {$\state_1$};
    \node at (0,2) (q2) {$\state_\channelmessage$};
    \node at (0,3) (q3) {$\state_3$};
    \node at (0,4) (q4) {$\state_4$};
    \node at (0,5) (q5) {$\state_5$};
    \node at (0,6) (q6) {$\state_6$};
    \node at (0,7) (q7) {$\state_7$};
    \node at (0,8) (q8) {$\state_8$};
    \node at (0,9) (q9) {$\state_9$};
    \node at (0,10) (q10) {$\state_{10}$};
    \node at (0,11) (q11) {$\state_{11}$};
    \node at (0,12) (q1') {$\state_1$};

    \draw[->] (q1) -- node[right] {$\rd(\xwr,\channelmessage)$} (q2);
    \draw[->] (q2) -- node[right] {$\wr(\xrd,\channelmessage)$} (q3);
    \draw[->] (q3) -- node[right] {$\wr(\xwr,\bot)$} (q4);
    \draw[->] (q4) -- node[right] {$\mf$} (q5);
    \draw[->] (q5) -- node[right] {$\rd(\yvar,0)$} (q6);
    \draw[->] (q6) -- node[right] {$\rd(\yvar,1)$} (q7);
    \draw[->] (q7) -- node[right] {$\wr(\yvar,0)$} (q8);
    \draw[->] (q8) -- node[right] {$\mf$} (q9);
    \draw[->] (q9) -- node[right] {$\rd(\xwr,\bot)$} (q10);
    \draw[->] (q10) -- node[right] {$\wr(\xrd,\bot)$} (q11);
    \draw[->] (q11) -- node[right] {$\mf$} (q1');

    \node at (5,5) (qf) {$\state_F$};
    \draw[->] (q5) -- node[above right] {$\rd(\yvar,1)$} (qf);

    \node at (5,6) (qf) {$\state_F$};
    \draw[->] (q6) -- node[above right] {$\rd(\yvar,0)$} (qf);

    \node at (5,9) (qf) {$\state_F$};
    \draw[->] (q9) -- node[above right] {$\rd(\xwr,\channelmessage)$} (qf);

    \draw[thick,decoration={brace, mirror},decorate] (q1.south west) -- node[left]{for all $\channelmessage \in \channelmessageset$} (q3.north west);
\end{tikzpicture}
\caption{$\process^2$ of the reduction from PCS.}
\label{fig:ru-reduction-2}
\end{minipage}
\end{figure}

The first process keeps track of the configuration of the channel system and simulates the control flow.
For each transition in $\channelsystem$, we construct a sequence of transitions in $\process^1$ that simulates both the state change and the channel behaviour of the $\channelsystem$-transition.
To achieve this, $\process^1$ uses its buffer to store the messages of the PCS's channel.
In particular, to simulate a send operation $!\channelmessage$, $\process^1$ adds the message $\tuple{\xwr, \channelmessage}$ to its buffer.
For receive operations, $\process^1$ cannot read its own oldest buffer message, since it is overshadowed by the more recent messages.
Thus, the program uses the second process $\process^2$ to read the message from memory and copies it to the variable $\xrd$, where it can be read by $\process^1$.
We call the combination of reading a message $\channelmessage$ from $\xwr$ and writing it to $\xrd$ the \emph{rotation} of $\channelmessage$.

While this is sufficient to simulate all behaviours of the PCS, it also allows for additional behaviour that is not captured by $\channelsystem$.
More precisely, we need to ensure that each channel message is received \emph{once and only once}.
Equivalently, we need to prevent the \emph{loss} and \emph{duplication} of messages.
This can happen due to multiple reasons.

First, the update player might choose to lose a channel message by updating more than one message during a rotation.
Consider an execution of $\program$ that simulates two send operations $!\channelmessage_1$ and $!\channelmessage_2$, i.e. $\process^1$ adds $\tuple{\xwr, \channelmessage_1}$ and $\tuple{\xwr, \channelmessage_2}$ to its buffer.
Now, if the process player wants to simulate a receive operation and initiates a message rotation, the update player can update both messages $\tuple{\xwr, \channelmessage_1}$ and $\tuple{\xwr, \channelmessage_2}$ to the memory before $\process^2$ reads from $\xwr$.
Thus, the first message $\channelmessage_1$ is overwritten by the second message $\channelmessage_2$ and is lost without ever being received.
To prevent this, we implement a protocol that ensures that in each message rotation, exactly one channel message is being updated.

We extend the construction of $\process^1$ such that it inserts an auxiliary message $\tuple{\yvar, 1}$ into its buffer after the simulation of each send operation.
After a message rotation, that is, after $\process^2$ copied a message from $\xwr$ to $\xrd$, the process then resets the value of $\xwr$ to its initial value $\bot$.
Next, the process checks that $\yvar$ contains the value $0$, which indicates that only one message was updated to the memory.
Now, the update player is allowed to update exactly one $\tuple{\yvar, 1}$ buffer message, after which $\process^2$ resets $\yvar$ to $0$.
To ensure that the update player has actually updated only one message in this step, $\process^2$ then checks that $\xwr$ is still empty.
If this protocol is violated at any point, $\process^2$ enables the process player to immediately move to a winning state.

Although we have established that during each message rotation exactly one channel message will be rotated, we also need to ensure that for each rotation, $\process^1$ will simulate exactly one receive operation.
This is achieved by another protocol between $\process^1$ and $\process^2$, which gives the update player the tools to enforce correct behaviour.
To begin, the process player needs to initiate the simulation of a receive operation by moving to the first auxiliary state $\hstate_1$ shown in \autoref{fig:ru-reduction-1c}.
Only then is the program in a deadlock and the update player is forced to perform a message update.
When reaching the first memory fence in $\process^2$, the system is deadlocked again.
Of course, the update player will not update the next message in the buffer of $\process^1$, since it will lead to the process player immediately winning later on.
Thus, she updates the message $\tuple{\xrd, \channelmessage}$ to the memory, which enables both processes to continue.
The next time that the update player is forced to update is when $\process^1$ reaches the second auxiliary state $\hstate_2$ and $\process^2$ reaches the second memory fence.
Only emptying the buffer of $\process^2$ allows the program to continue.
After three more instructions, $\process^2$ will reach the third memory fence.
Again, the update player needs to empty the buffer of $\process^2$ which updates $\tuple{\xrd, \bot}$ and enables $\process^1$ to finish the simulation of the receive operation.

This concludes the mechanisms implemented to ensure that each channel message is received \emph{once and only once}.
We have constructed a TSO game with update fairness that simulates a perfect channel system.
The winning condition of the game will be the reachability condition induced by the final states of the PCS together with the update fairness condition.
We summarise our results in the following theorem.

\begin{theorem}
    The reachability problem for TSO games with update fairness is undecidable.
\end{theorem}

\section{Process Fairness in Safety Games}
\label{sec:process-fairness}

In the previous section, we have limited the behaviour of the update player.
Now, we introduce \emph{process fairness}, which will impact the capabilities of the process player.
Process fairness means that for each process that is enabled infinitely many times during a run, the process player executes an instruction in that process infinitely often.
In reachability games, this is no real restriction to the process player:
If she can reach the set of winning states, then she can do so in finitely many moves.
Thus, any finite prefix of a play that reaches a winning state can then trivially be extended to an infinite play that admits process fairness.
Because of this, we will target our attention only towards safety games, where the process player cannot win in a finite amount of moves.

We formalise process fairness as follows.
Let $\program = \tuple{\process^\pid}_{\pid\in\indexset}$ be a TSO program with induced game $\game\of\program$.
We define $\wincon_P$ as the set of all plays $\play = \conf_0, \conf_1, \dots$ in $\game\of\program$ that satisfy process fairness:
$$ \play \in \wincon_P \iff \forall\ \pid \in \indexset: \left(
\exists^\infty\ k \in \Nat, \conf_k \in \confset_A: \conf_k \to[\instr_\pid] \conf'
\implies
\exists^\infty\ k' \in \Nat, \conf_{k'} \in \confset_A: \conf_{k'} \to[\instr'_\pid] \conf_{k'+1}
\right) $$
Given a safety condition $\wincon_S$, the intersection $\wincon_{SP} = \wincon_S \cap \wincon_P$ defines the set of winning plays that admit process fairness.
In the remainder of this section we will show that the safety problem under process fairness is undecidable.
To do so, we use a construction very similar to the one from the previous section to reduce the state reachability problem of perfect channel systems, to the safety problem of $\game\of\program$.
Before, it was the process player who decided which transition of the channel system to simulate.
This time, it will be the update player who has this task.

Consider again a perfect channel system $\channelsystem = \tuple{ \channelstateset, \channelmessageset, \transition }$.
We modify the construction from the previous section.
First, we introduce another shared variable $\zvar$.
Then, for each transition $\channeledge \in \transition$ of the perfect channel system, we add an auxiliary process $\process^\channeledge$.
It consists of exactly one state $\state_\channeledge$ and one looping transition $\state_\channeledge \to[\wr(\zvar, \channeledge)] \state_\channeledge$.
Furthermore, we prepend the gadget of process $\process^1$ that simulates $\channeledge$ with a transition $\rd(\zvar, \channeledge)$.
The result of this is shown in \autoref{fig:sp-reduction-1}.
Process $\process^2$ is taken from the previous construction without any changes and can be found in \autoref{fig:ru-reduction-2}.

\begin{figure}
\centering

\begin{subfigure}[b]{0.3\linewidth}
\centering
\begin{tikzpicture}[yscale=-1]
    \node at (0,-1) {};
    \node at (0,0) (q1) {$\channelstate$};
    \node at (0,1) (q2) {$\channelstate'$};
    \node[transparent] at (0,2) {$\channelstate$};

    \draw[->] (q1) -- node[right] {$\rd(\zvar, \channeledge)$} (q2);
\end{tikzpicture}
\bigskip
\multicaption{skip operation}{$\channeledge = \channelstate \to[\nop]_\channelsystem \channelstate' \in \transition$}
\end{subfigure}
\hfill
%
\begin{subfigure}[b]{0.3\linewidth}
\centering
\begin{tikzpicture}[yscale=-1]
    \node at (0,0) (q1) {$\channelstate$};
    \node at (0,1) (h1) {$\hstate_1$};
    \node at (0,2) (h2) {$\hstate_2$};
    \node at (0,3) (q2) {$\channelstate'$};

    \draw[->] (q1) -- node[right] {$\rd(\zvar, \channeledge)$} (h1);
    \draw[->] (h1) -- node[right] {$\wr(\xwr,\channelmessage)$} (h2);
    \draw[->] (h2) -- node[right] {$\wr(\yvar,1)$} (q2);
\end{tikzpicture}
\bigskip
\multicaption{send operation}{$\channeledge = \channelstate \to[!\channelmessage]_\channelsystem \channelstate' \in \transition$}
\end{subfigure}
\hfill
%
\begin{subfigure}[b]{0.3\linewidth}
\centering
\begin{tikzpicture}[yscale=-1]
    \node at (0,0) (q1) {$\channelstate$};
    \node at (0,1) (h1) {$\hstate_1$};
    \node at (0,2) (h2) {$\hstate_2$};
    \node at (0,3) (q2) {$\channelstate'$};

    \draw[->] (q1) -- node[right] {$\rd(\zvar, \channeledge)$} (h1);
    \draw[->] (h1) -- node[right] {$\rd(\xrd,\channelmessage)$} (h2);
    \draw[->] (h2) -- node[right] {$\rd(\xrd,\bot)$} (q2);
\end{tikzpicture}
\bigskip
\multicaption{receive operation}{$\channeledge = \channelstate \to[?\channelmessage]_\channelsystem \channelstate' \in \transition$}
\label{fig:sp-reduction-1c}
\end{subfigure}

\caption{$\process^1$ of the reduction from PCS to a TSO game with process fairness.}
\label{fig:sp-reduction-1}
\end{figure}
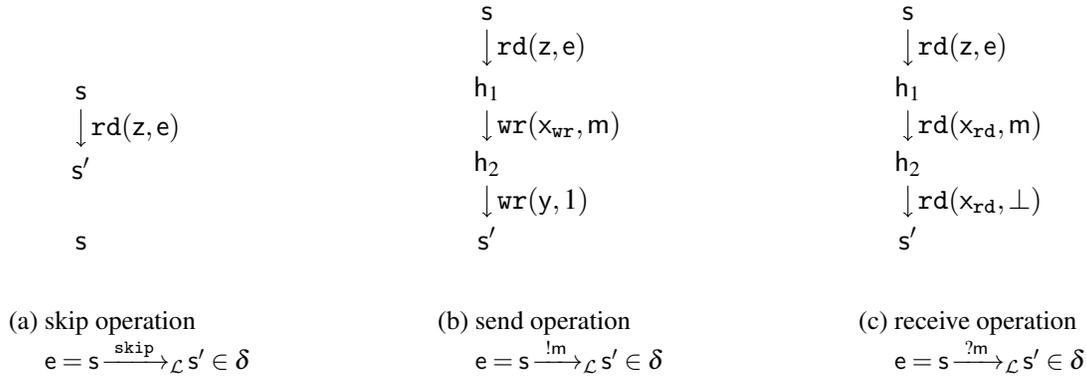

The main idea of these modifications is that the update player can use the variable $\zvar$ to control which channel operation will be simulated.
At the start of the run, both $\zvar$ and $\xwr$ contain the initial value $\bot$, which means that neither $\process^1$ nor $\process^2$ are enabled.
Thus, the process player needs to begin playing in some process $\process^\channeledge$, writing the message $\tuple{\zvar, \channeledge}$ to its buffer.
This will continue until the update player decides to update one of these messages.
But, due to process fairness, the process player is forced to eventually play in \emph{all} enabled processes.
In particular, she has to do so infinitely many times during any infinite play.
This means that the update player can simply wait until each transition $\channeledge \in \transition$ was sufficiently many times added to the buffer of $\process^\channeledge$ to simulate a run of the PCS that reaches a final state.
At that point, the update player starts updating the messages $\tuple{\zvar, \channeledge}$ one by one, each time waiting until the process player has finished simulating the unique  operation that is enabled in $\process^1$.

In more detail, to simulate the execution of a channel operation $\channeledge \in \transition$, the update player updates the buffer message $\tuple{\zvar, \channeledge}$ to the memory.
Due to process fairness, we know that the process player eventually has to play in $\process^1$, since it is now enabled.
She takes the transition $\rd(\zvar, \channeledge)$ and then proceeds (although not necessarily immediately) with simulating $\channeledge$ as was presented for the reachability case in the previous section.
If $\channeledge$ is a receive operation, the update player has to update a $\tuple{\xwr, \channelmessage}$ buffer message at some point to enable $\process^2$ and start a message rotation.
Again, process fairness forces the process player to eventually finish the rotation protocol.
In any case, the simulation of the channel operation is guaranteed to terminate after finitely many steps.
Now, the update player starts the next simulation by updating the corresponding buffer message.

Also in this construction we need to ensure that we do not introduce any behaviour that does not correspond exactly to what the PCS can do.
Message loss due to updating two messages without rotation in between is handled in the same way as previously, using the auxiliary variable $\yvar$.
The same goes for message duplication, which is covered by the protocol between $\process^1$ and $\process^2$.
What is left is message loss due to performing two rotations without simulating a receive operation.
In the previous section, this could only happen if the process player decides to do so, since she was the one controlling the simulation.
This is still prevented by the aforementioned protocol:
The update player is not forced to let $\process^2$ proceed beyond the second and third memory fences before $\process^1$ keeps up with the protocol.
But in this construction, the roles and capabilities of the two players have changed slightly and allow for additional behaviour:
Without enabling a receive operation in $\process^1$, the update player could update a message $\tuple{\xwr, \channelmessage}$, which enables $\process^2$ instead.
Due to process fairness, the process player would eventually have to perform a full message rotation.

We prevent this by adding for every channel message $\channelmessage$ a transition $\state \to[\rd(\xrd, \channelmessage)] \state_F$ to \emph{every} state $\state$ of $\process^1$ \emph{except} the states $\hstate_1$ and $\hstate_2$ (cf. \autoref{fig:sp-reduction-1c}) of the receive operations of message $\channelmessage$.
Here, $\state_F$ is a sink state which is safe for the process player and blocks the update player from winning the game.
The idea of this transition is that during a message rotation, the value of $\xrd$ in the memory is $\channelmessage$.
Thus, the update player is not immediately losing only if $\process^1$ is currently in one of the intermediate states of a receive operation.
In particular, the process player can now move to $\hstate_2$.
Then, after the rotation has finished with the third memory fence, the variable $\xrd$ contains the value $\bot$ again, which means that $\process^1$ is enabled and the process player can finish the simulation of the receive operation.
We conclude that she can prevent the update player from performing a rotation without simulating a receive operation.

In summary, we have shown again that each channel message is read once and only once.
The winning condition of the TSO game is given by the safety condition induced by the final states of the PCS together with the process fairness condition.
This gives rise to the following theorem.
\begin{theorem}
    The safety problem for TSO games with process fairness is undecidable.
\end{theorem}

\section{Load Buffer Semantics in TSO Games}

In \cite{DBLP:journals/lmcs/AbdullaABN18}, the authors introduced an alternative semantics for TSO, called \emph{load buffer semantics}.
It is equivalent to the traditional \emph{store buffer semantics} in the sense that a global state $\statemap$ of the system is reachable under load-buffer semantics if and only if it is reachable under store buffer semantics.
The alternative semantics have been proven to be useful in efficiently performing algorithmic verification or presenting simpler decidability proofs of safety properties.
A natural question in the context of this paper is to ask whether these results transfer to the game setting.
In particular, we want to know if a game is won by the same player when played under both semantics.
Unfortunately, it turns out that this is not the case.

\paragraph{Load Buffer Semantics}
Under the new semantics, the store buffer between each process and the shared memory is replaced by a load buffer instead.
This means that the information flow reverses its direction:
Instead of write operations, the buffer now contains potential \emph{read} operations that \emph{might} be performed by the process.
Each buffer message is either a pair $\tuple\xd$ or a triple $\tuple{\xd,\own}$, where the latter is called an \emph{own-message}.

At any point during the run, the system can nondeterministically choose a variable $\xvar$ and its corresponding value $\dval$ from the memory and add a message $\tuple\xd$ to the tail of the buffer of one of the processes.
This is called \emph{read propagation} and speculates on a future read operation on $\xvar$.
Conversely, a \emph{delete} operation removes the oldest message at the head of the buffer of some process and is also performed nondeterministically at any time.

A \emph{write} instruction $\wr\of\xd$ of a process $\process$ immediately updates the value $\dval$ of the variable $\xvar$ in the memory.
Then, it adds the own-message $\tuple{\xd,\own}$ to the buffer of $\process$.
The behaviour of a \emph{read} instruction $\rd\of\xd$ depends on the contents of the buffer.
If there is an own-message on the variable $\xvar$, then the most recent one must correspond to the value $\dval$.
Otherwise, if there is no such message, the head of the buffer must be a message $\tuple{\xd}$.
If this is not the case, the read instruction is disabled.
The last two instructions, which are \emph{skip} and \emph{memory fence}, work exactly as in the classical TSO semantics:
They only change the local state but not the memory or buffer, and the fence is only enabled if its buffer is empty.

For a formal definition of the semantics and the configurations of the induced transition system, we refer to \cite{DBLP:journals/lmcs/AbdullaABN18}.

\paragraph{Games}
Comparing the alternative to the classical TSO semantics, we see that the order in which the variables in the memory are updated, now depends directly on the order of execution of the corresponding write instructions, and not on the update order of the buffer messages.
Conversely, the buffer does not delay the time when a write operation arrives at the memory, but instead delays when the change in the memory is visible to each process.
Intuitively, we can already guess that the two semantics differ in the game setting, since the information available to the two players during the execution is different in both cases.

In the plain TSO game without fairness, there is actually no change.
Since both reachability and safety games degenerate to single-process games with no communication between processes, the exact update semantics do not matter.

\begin{figure}
\centering

\begin{subfigure}[b]{0.2\linewidth}
\centering
\begin{tikzpicture}[yscale=-1]
    \node at (0,-1) {};
    \node at (0,0) (q1) {$\state_1$};
    \node at (0,1) (q2) {$\state_2$};
    \node[transparent] at (0,2) {$\state$};

    \draw[->] (q1) -- node[right] {$\wr(\xvar, 1)$} (q2);
\end{tikzpicture}
\bigskip
\caption{$\process^1$}
\end{subfigure}
\hfill
\begin{subfigure}[b]{0.2\linewidth}
\centering
\begin{tikzpicture}[yscale=-1]
    \node at (0,-1) {};
    \node at (0,0) (q1) {$\state_1$};
    \node at (0,1) (q2) {$\state_2$};
    \node[transparent] at (0,2) {$\state$};

    \draw[->] (q1) -- node[right] {$\wr(\xvar, 2)$} (q2);
\end{tikzpicture}
\bigskip
\caption{$\process^2$}
\end{subfigure}
\hfill
\begin{subfigure}[b]{0.5\linewidth}
\centering
\begin{tikzpicture}[xscale=1.5,yscale=-2]
    \node at ( 0,0) (q1) {$\state_1$};
    \node at (-1,1) (l1) {$\state_2$};
    \node at ( 1,1) (r1) {$\state_3$};
    \node at (-1,2) (l2) {$\state_4$};
    \node at ( 1,2) (r2) {$\state_5$};
    \node at ( 0,3) (q2) {$\state_F$};

    \draw[->] (q1) -- node[above left]  {$\nop$} (l1);
    \draw[->] (q1) -- node[above right] {$\nop$} (r1);
    \draw[->] (l1) to[out=180,in=240,looseness=8] node[left ,align=left] {$\rd(\xvar,0)$\\$\rd(\xvar,2)$} (l1);
    \draw[->] (r1) to[out=  0,in=-60,looseness=8] node[right,align=left] {$\rd(\xvar,0)$\\$\rd(\xvar,1)$} (r1);
    \draw[->] (l1) -- node[left]  {$\rd(\xvar,1)$} (l2);
    \draw[->] (r1) -- node[right] {$\rd(\xvar,2)$} (r2);
    \draw[->] (l2) to[out=120,in=180,looseness=8] node[left ] {$\rd(\xvar,1)$} (l2);
    \draw[->] (r2) to[out= 60,in=  0,looseness=8] node[right] {$\rd(\xvar,2)$} (r2);
    \draw[->] (l2) -- node[below left]  {$\rd(\xvar,2)$} (q2);
    \draw[->] (r2) -- node[below right] {$\rd(\xvar,1)$} (q2);
\end{tikzpicture}
\bigskip
\caption{$\process^3$}
\end{subfigure}

\caption{A concurrent program consisting of three processes.}
\label{fig:load-buffer-sp}
\end{figure}
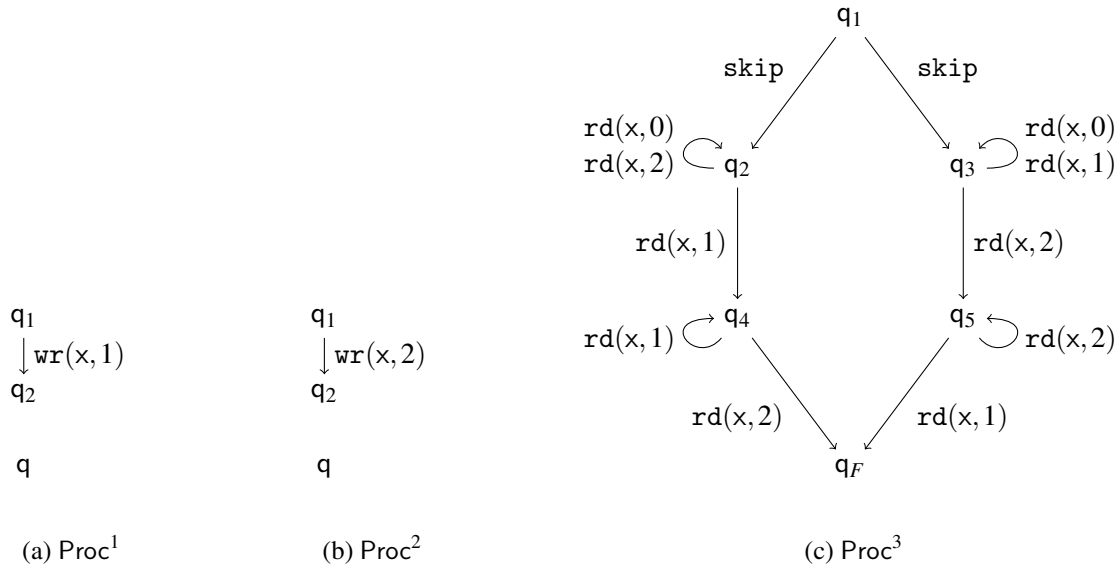

This is not the case if we add fairness conditions.
Consider the safety game with process fairness played on the program shown in \autoref{fig:load-buffer-sp}.
The target state is $\state_F$ in $\process^3$ and the initial value of $\xvar$ is $0$.
Since the process player cannot be deadlocked in any other state of $\process^3$, the only way for the update player to win is to force the play into $\state_F$.
In our classical game setting under store buffer semantics, the update player is able to achieve this.
We describe a winning strategy for her.

Due to process fairness, the process player needs to eventually play in all three processes.
The update player waits until processes $\process^1$ and $\process^2$ are in their respective $\state_2$, and $\process_3$ is in either $\state_2$ or $\state_3$.
In the first case, if $\process_3$ is in state $\state_2$, she updates the buffer message of $\process^1$, which writes the value $1$ to the memory for variable $\xvar$.
The only enabled instruction for the process player is to move from $\state_2$ to $\state_4$ in $\process^3$.
Now, the update player updates the message from the buffer of $\process^2$, which forces the process player to move to the target state and lose.
In the other case, the update player performs the two update operations in the reverse order, which again forces the process player to enter the target state after two moves.

Next, consider the same game but under load buffer semantics.
We have not yet formally defined how they should work, but this is not necessary for our argument.
Assume that the process player as usual controls the program instructions and the update player in some way controls the nondeterministic buffer behaviour.
We will outline how the process player wins this game.

First, she plays in $\process^1$, then in $\process^2$.
At this point, the value of $\xvar$ in the memory is $2$, but the buffer of $\process^3$ might already contain messages of the form $\tuple{\xvar, 1}$ and $\tuple{\xvar, 2}$.
Note that it is only possible to have them in this exact order, i.e. it cannot be that there is some message $\tuple{\xvar, 2}$ that is older than another message $\tuple{\xvar, 1}$.
Furthermore, since the program has no other reachable write instructions, any message that will be added in the future must be $\tuple{\xvar, 2}$.
Now, the process player plays in $\process^3$ and moves to $\state_3$.
The update player needs the process player to eventually move to $\state_5$, which means she has to enable the instruction $\rd(\xvar, 2)$.
To do so, she deletes messages at the head of the buffer of $\process_3$ until it reaches a message $\tuple{\xvar, 2}$.
But due to the order of the messages in the buffer, this means that it lost all messages $\tuple{\xvar, 1}$ and also, as said previously, cannot add any more of them.
It follows that the process can never execute the next instruction $\rd(\xvar, 1)$ and is thus stuck in $\state_5$.
Since this is not a deadlock for the process player, it results in a winning play for her.

We conclude that TSO safety games with process fairness do not have the same winning configurations under store buffer semantics and load buffer semantics, respectively.
The same can be shown for reachability games with update fairness.
Since it does not yield any additional insights, we do not present the argument here.

\section{Conclusion and Future Work}
In this paper, we continue the work on two-player games played on programs running under TSO semantics.
We present a game model where one player controls the instructions of the program and the other player controls the buffer updates.
Our results show that both the reachability problem and the safety problem for these games reduce to the analysis of games on single-process programs.
Moreover, we show a bisimilarity to a game with a finite amount of configurations and use it to prove that the problems are in fact \pspace-complete.

The reduced complexity comes from the optimal behaviour of the two players.
The process player can always win by playing in only one single process, while the best strategy of the update player is to stay passive and not perform any buffer updates.
We rectify this by introducing fairness conditions for both players.
In reachability games, the update player is required to update each message eventually.
This allows the process player to wait for a write instruction to arrive in the memory.
In safety games, during an infinite run the process player has to perform instructions in all enabled processes infinitely often.
Both restrictions lead to the respective problems being undecidable.

Finally, we connect the game model to the alternative load buffer semantics of TSO.
We show that the equivalence between load buffer and store buffer that exists for classical TSO reachability does not carry over to the game setting.

\bigskip

This work analyses the basic winning conditions reachability and safety.
Future work may expand the focus to more expressive winning conditions, like Büchi (i.e. repeated reachability), Co-Büchi, Parity, Rabin, Streett or Muller.
Another way to expand is to look at other fairness conditions for the two players, for example transition fairness.
These two directions of research are not orthogonal to each other since Muller or even Streett conditions might be able to encode some forms of fairness conditions.

\newpage
\bibliographystyle{eptcs}
\bibliography{bibdatabase}

\end{document}